\documentclass[12pt]{article}
\usepackage{titlesec,titletoc}
\usepackage[mode=buildmissing]{standalone}

 \usepackage{geometry}

\usepackage[
			anchorcolor=black,
			citecolor=black,
			colorlinks=true,
			filecolor=black,
			linkcolor=black,
			menucolor=black,
			runcolor=black,
			urlcolor=black,
			hyperfootnotes=false
]{hyperref}
\usepackage{amsthm,amsmath,microtype}
\usepackage[noabbrev]{cleveref}

\theoremstyle{plain}

\newtheorem{lemma}{Lemma}
\newtheorem{proposition}{Proposition}
\newtheorem{corollary}{Corollary}
\theoremstyle{definition}
\newtheorem{definition}{Definition}

\theoremstyle{remark}

\crefname{observation}{Observation}{Observations}
\crefname{theorem}{Theorem}{Theorems}
\crefname{lemma}{Lemma}{Lemmata}
\crefname{proposition}{Proposition}{Propositions}
\crefname{corollary}{Corollary}{Corollaries}
\crefname{definition}{Definition}{Definitions}
\crefname{assumption}{Assumption}{Assumptions}
\crefname{section}{Section}{Sections}
\crefname{figure}{Figure}{Figures}
\crefname{axiom}{Axiom}{Axioms}
\crefname{modification}{Modification}{Modifications}
\crefname{example}{Example}{Examples}
\crefname{appendix}{Appendix}{Appendices}

\usepackage[font={small}]{caption}[2008/08/24] 

\usepackage{enumitem} 

\usepackage{setspace}

\usepackage{amsmath,amsfonts,amssymb, mathrsfs}
\usepackage[T1]{fontenc}
\usepackage[]{lmodern}
\usepackage[utf8]{inputenc}
\usepackage[ngerman, english]{babel}
\usepackage{csquotes} 
\usepackage{etex} 
\usepackage[subrefformat=parens,labelformat=parens]{subfig}
\usepackage{multirow,array}

\usepackage[noabbrev]{cleveref}

\usepackage{pgfplots}
\pgfplotsset{compat=newest}
\usepackage{sidecap} 
\usetikzlibrary{patterns,decorations.pathreplacing,positioning,external,shapes.geometric,arrows,calc}

\usepackage{pstricks} 
\usepackage{pst-3d}
\usepackage{sgamevar, egameps} 
\usepackage{graphicx}
\usepackage{epstopdf} 

\usepackage[backend=biber,
			sorting=ynt,
			sortcites=true,
			bibencoding=inputenc,
			bibstyle=authoryear -comp,
			citestyle=authoryear -comp,
			doi=false,
			giveninits=true,
			isbn=false,
			uniquename=false,
			uniquelist=minyear,
			maxcitenames=2,
			maxnames=2,
			maxbibnames=99,
			natbib=true,
			uniquename=false,
			url=false,
			date=year
]{biblatex}
\bibliography{Latex_templates/mybib}

\DeclareSourcemap{
  \maps[datatype=bibtex]{
    \map[overwrite=true]{
      \step[fieldsource=number]
      \step[fieldset=issue, origfieldval]
    }
  }
}

	\AtEveryBibitem{\clearlist{language}}
	\AtEveryBibitem{\clearfield{number}}
	\AtEveryBibitem{\clearlist{month}}
	\AtEveryBibitem{\clearfield{Month}}
	\AtEveryBibitem{\clearfield{day}}
	\AtEveryBibitem{\clearfield{eprint}}

\renewbibmacro{in:}{}
\usepackage{verbatim}
\usepackage{booktabs} 
\usepackage[multiple]{footmisc}

\usepackage{placeins} 

\definecolor{Vermilion}{RGB}{213,94,0}
\definecolor{BluishGreen}{RGB}{0,158,115}
\definecolor{Cadmium}{RGB}{227, 0, 34}
\definecolor{ReddishPurple}{RGB}{204,121,167}
\definecolor{Orange}{RGB}{230,159,0}
\definecolor{Yellow}{RGB}{240,228,66}
\definecolor{SkyBlue}{RGB}{86,180,233}
\definecolor{Blue}{RGB}{0,114,178}

\DeclareFontFamily{U}{mathx}{\hyphenchar\font45}
\DeclareFontShape{U}{mathx}{m}{n}{
      <5> <6> <7> <8> <9> <10>
      <10.95> <12> <14.4> <17.28> <20.74> <24.88>
      mathx10
      }{}
\DeclareSymbolFont{mathx}{U}{mathx}{m}{n}
\DeclareFontSubstitution{U}{mathx}{m}{n}
\DeclareMathAccent{\widecheck}{0}{mathx}{"71}
\DeclareMathAccent{\wideparen}{0}{mathx}{"75}

\DeclareSourcemap{
  \maps[datatype=bibtex,overwrite=true]{
    \map{
      \step[fieldsource=pages,
            match=\regexp{pp?\.?(.+)},
           replace=\regexp{$1}] 
    }
  }
}

\makeatletter
\g@addto@macro\normalsize{%
  \setlength\abovedisplayskip{3pt}
  \setlength\belowdisplayskip{8pt}
  \setlength\abovedisplayshortskip{3pt}
  \setlength\belowdisplayshortskip{8pt}
}
\makeatother

\titlespacing*{\subsection}
{0pt}{0.8ex plus 1ex minus .2ex}{1.2ex plus .2ex}
\titlespacing*{\section}
{0pt}{0.8ex plus 1ex minus .2ex}{1.2ex plus .2ex}
\titlespacing*{\paragraph}
{0pt}{0.5ex plus 1ex minus .2ex}{0.5ex plus .2ex}

\usepackage{datetime}
\usepackage[commandnameprefix=ifneeded,final]{changes}
\setlist{noitemsep,topsep=0.5ex}
\assignrefcontextentries[]{*}

\theoremstyle{definition}
\newtheorem{property}{Property}
\theoremstyle{plain}

\title{A Quest for Knowledge\thanks{We thank Arjada Bardhi, Antonio Cabrales, Steven Callander, Marco Celentani, Rahul Deb, Philipp Denter, Florian Ederer, Chiara Franzoni, Alex Frug, William Fuchs, Ben Golub, Ryan Hill, Toomas Hinosaar, Nenad Kos, Jorge Lemus, Igor Letina, David Lindequist, Gerard Llobet, Ignacio Ortu\~no, Marco Ottaviani, Nicola Pavoni, Harry Pei, Konrad Stahl, Carolyn Stein, Armin Schmutzler, Carlo Schwarz, Adrien Vigier, Ludo Visschers, Ralph Winkler, and various audiences for comments. \newline
Johannes Schneider gratefully acknowledges financial support from Horizon Europe MSCA Project 101061192; the Agencia Estatal de Investigación (MICIU/AEI /10.13039/501100011033) through grants PID2019-111095RB-I00, PID2020-118022GB-I00, and CEX2021-001181-M; the German Research Foundation (DFG) through CRC TR 224 (Project B03); and Comunidad de Madrid, grant MAD-ECON-POL-CM H2019/HUM-5891. Part of the research was carried out while Johannes Schneider was visiting the University of Bern. He thanks them for their hospitality and inspiration.}}

\makeatletter
\def\pgfplots@stacked@diff{}
\makeatother

\author{Christoph Carnehl\thanks{Bocconi University, Dpt. of Economics \& IGIER; email: \nolinkurl{christoph.carnehl@unibocconi.it}} \and Johannes Schneider\thanks{Universidad Carlos III de Madrid; Department of Economics, email: \nolinkurl{jschneid@econ.uc3m.es}} 
}

\date{\monthname~\the\year}
\tikzset{font=\footnotesize}
\pgfplotsset{%
    ,tick label style = {font=\footnotesize}
    ,every axis label = {font=\footnotesize}
    ,legend style = {font=\footnotesize}
    ,label style = {font=\footnotesize}
}
\usepgfplotslibrary{fillbetween}
\begin{document}


\captionsetup{belowskip=0pt}
\maketitle
\spacing{1.25}
\begin{abstract}
Is more novel research always desirable? We develop a model in which knowledge shapes society's policies and guides the search for discoveries. Researchers select a question and how intensely to study it. The novelty of a question determines both the value and difficulty of discovering its answer. We show that the benefits of discoveries are nonmonotone in novelty. Knowledge expands endogenously step-by-step over time. Through a dynamic externality, moonshots---research on questions more novel than what is myopically optimal---can improve the evolution of knowledge. Moonshots induce research cycles in which subsequent researchers connect the moonshot to previous knowledge.
\end{abstract}

\newpage

\section{Introduction} 
\label{sec:introduction}

In a letter to Franklin D. Roosevelt, Vannevar \citet{Bush45} pleads with the president to preserve freedom of inquiry by federally funding basic research\textemdash the ``pacemaker of technological progress.'' That letter paved the way for the creation of the National Science Foundation (NSF) in 1950. The NSF today, like the vast majority of governments and scientific institutions, cherishes scientific freedom and allows academic researchers to select research projects independently.

With scientific freedom comes the responsibility to select the right research questions. However, what are the right questions? Biologist and Nobel laureate Peter \citet{Medawar1967-MEDTAO-5} famously notes that ``research is surely the art of the soluble. \ldots Good scientists study the most important problems they think they can solve.'' Finding the most important yet soluble question is nontrivial. One reason is that both importance and solubility depend on the current state of knowledge. 

In this paper, we develop a microfounded model of knowledge creation through research. We conceptualize research as a costly search process that may fail. The researcher picks a question and an intensity with which to search for its answer. The cost of search depends on how well existing knowledge guides the researcher's efforts. If the researcher discovers the answer, her gross benefits depend on how much society's decision-making improves due to the discovery---both through answering a particular question and the spillovers on related questions. The discovery is then added to the stock of knowledge and future researchers can build on the additional knowledge. We characterize the researcher's optimal choice for arbitrary existing knowledge, and find that researchers work on questions that are too narrow and fail too often. The questions chosen are neither novel enough to maximize the instantaneous value of knowledge for decision-making nor novel enough to inspire future generations. By incentivizing distant discoveries, a designer can reduce the researchers' failure rate and improve the evolution of knowledge.

We model the value of knowledge as the extent to which knowledge improves decision-making. We represent society by a single decision-maker who can use existing knowledge to address a variety of problems. Existing knowledge is the set of questions to which the answer has already been discovered. Because answers to similar questions are correlated, knowledge also provides the decision-maker with conjectures regarding questions to which the answer is yet undiscovered. The precision of a conjecture depends on the question's location relative to existing knowledge.\footnote{The protein folding problem in structural biology provides a case in point. Spillovers from other proteins led to Moderna's development of the COVID-19 vaccine, which ``took all of one weekend'' (only). For more on the protein folding problem see \citet{hill2020race,hill2019scooped}.} 

We conceptualize the correlation by assuming that answers to questions follow the realization of a Brownian path. \cref{fig:Intro} depicts that idea. Questions are on the horizontal axis, and the gray line represents the answers to all questions. Dots (\tikz{\draw[fill=red,red]  circle(.5ex)}) represent existing knowledge. Because of the assumption of a Brownian path, all conjectures follow a normal distribution. The mean and the variance depend on existing knowledge. The solid black lines in \cref{fig:Intro} represent the mean; the dashed lines provide the band of the 95\% prediction interval.\footnote{The 95\% prediction intervals describe the following relation: for each question, with a probability of 95\%, the answer lies between the respective dashed lines given existing knowledge. For visual clarity, we selected a Brownian path that leaves the 95$\%$ prediction interval in the negative branch when we condition on it passing through the origin only.} 

\begin{figure}[t]
  \subfloat{\includestandalone[width=.45\textwidth]{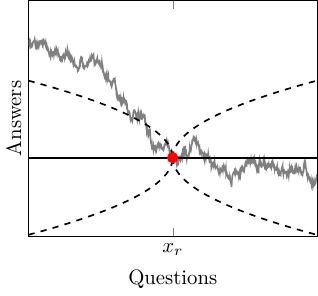}}\hfill
\subfloat{
\includestandalone[width=.45\textwidth]{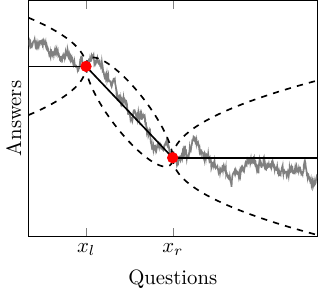}}
\caption{\emph{Existing knowledge and conjectures.}\label{fig:Intro}}
\end{figure}

Our first contribution is to characterize how adding a discovery improves the value of knowledge. To see what we have in mind, consider the left panel of \cref{fig:Intro}. Only the answer to question $x_r$ is known. Assume that a researcher discovers the answer to question $x_l$. Knowledge is now depicted in the right panel. Decision-making improves in three ways. First, the decision-maker has precise knowledge about the answer to $x_l$. Second, her conjectures about all questions to the left of $x_l$ improve. Third, her conjectures improve in the \deleted{newly created} area $[x_l,x_r]$. For all questions in that area, the decision-maker now has two pieces of knowledge that help her form the conjecture. 

How much decision-making benefits from discovering the answer to $x_l$ depends on $x_l$'s distance from $x_r$. The decision-maker cares about the range of questions she has a good conjecture about and the actual precision of each conjecture. Increasing the distance between $x_l$ and $x_r$ is similar to the effect of output expansion by a monopolist. Consider two alternatives for a discovery $x_l'<x_l$. Knowing $x'_l$ instead of $x_l$ implies that more questions lie inside the area $[x'_l,x_r]$ than in $[x_l,x_r]$---a marginal gain in the range of questions with precise conjectures. At the same time, the conjectures about each question in $[x_l,x_r]$ become less precise with $x_l'$---an inframarginal loss. The benefits of a discovery are maximized at an intermediate distance.

If both $x_l$ and $x_r$ are known initially, discoveries \emph{expand knowledge} beyond the frontier if the discoveries concern questions $x \notin [x_l,x_r]$. Expanding beyond the frontier works in the manner described in the previous paragraph. Alternatively, discoveries \emph{deepen knowledge} if they concern questions $x \in [x_l,x_r]$. Depending on the distance between $x_l$ and $x_r$, expanding knowledge or deepening knowledge may be optimal. If $x_l$ and $x_r$ are close, the conjecture about any question in $[x_l,x_r]$ is already precise. In this case, expanding knowledge provides larger benefits than deepening knowledge does. If $x_l$ and $x_r$ are far apart, conjectures about questions in $[x_l,x_r]$ are imprecise. Obtaining an answer to a question $x \in [x_l,x_r]$ divides this single area of imprecise conjectures into two areas with precise conjectures. In this case, deepening knowledge provides larger benefits than expanding knowledge. 

Overall, the decision-maker benefits most from deepening knowledge between distant, yet not too distant, pieces of knowledge. Expanding knowledge beyond the frontier beats deepening knowledge only if all available areas are short.

Our second contribution is to characterize a researcher's optimal choice of which question to tackle and how much effort to invest in studying that question for arbitrary existing knowledge. We assume that the researcher's benefits of a finding are proportional to the benefits for the decision-maker we have discussed above. In addition, we conceptualize the research process as the search for an answer. We assume that it requires effort to search for an answer and that the cost of effort is increasing and convex. \replaced{We propose a microfounded cost function that captures the following idea: For a given amount of effort, search is more likely to result in a discovery if the conjecture about the answer is more precise. Because conjectures depend on the distance to existing knowledge, our cost function naturally links \emph{novelty} (the distance to existing knowledge) and research \emph{output} (the probability that search results in a discovery).}{We propose a cost function that derives from this idea and provide a microfoundation. The cost function relates the cost of research to \emph{novelty} (the distance of the question from existing knowledge) and \emph{output} (the probability that search results in discovery). The link originates in the initial conjecture, which depends on the novelty of the question. The more precise that conjecture, the higher the output for any given level of effort.}

Regarding the researcher's choices, we show that novelty and output are nontrivially related. Depending on the structure of existing knowledge, the two can substitute or complement each other. Thus, in some cases selecting a more novel question actually decreases the risk of failure. Overall, the researcher creates more output when she deepens knowledge than when she expands beyond the frontier. Output peaks for deepening knowledge in areas of intermediate length. \replaced{This suggests that researchers prefer areas that contain a limited number of questions, even though they benefit from improving decision-making across a wide range of problems. They mainly do this for two reasons: First, discoveries in smaller areas lead to more significant improvements on each question in that area. Second, it is less costly to obtain a discovery.}{Discoveries in these areas also provide large benefits to society.} 

Our third contribution is to apply our previous insights to study the evolution of knowledge. We turn to a setting in which research is done sequentially and show that if a researcher expands knowledge, no future researcher will deepen knowledge. Therefore, the evolution of knowledge takes on a ladder structure. 

Starting from this observation, we study a simple policy intervention. Suppose a long-lived designer can direct the choices of one researcher. The designer aims to improve decision-making by enhancing the evolution of knowledge. Under natural conditions, the designer induces a research cycle by encouraging a moonshot discovery---a discovery far from existing knowledge. Moonshot discoveries are suboptimal in the short run. They create knowledge that is too disconnected from existing knowledge and therefore provides little immediate benefit. However, moonshots guide future researchers aiming at questions between the moonshot and previously existing knowledge. As a result of the moonshot, future researchers increase their output and pursue different questions. The knowledge they create becomes more valuable than otherwise. If the designer is patient and the cost of research is intermediate, the positive dynamic externality of moonshots dominates the implied myopic loss. \added{Thus, moonshots have the flavor of an infrastructure investment. By themselves, they provide little value, yet they enable future researchers to thrive by setting the stage.}

To summarize, we make three contributions. First, we offer a microfounded framework to study knowledge and research in a complex world. We quantify the value of a discovery when society extrapolates from knowledge to address a wide range of questions.
Second, we shed light on the nontrivial relation between the novelty of a question and the probability that a researcher discovers its answer. Novelty and the probability of discovery are endogenously linked through a microfounded cost function. Third, we provide novel insights into a classical question in the economics of science funding: should a funder incentivize research far beyond the frontier even if the immediate benefits of such a discovery are low? Yes, if the cost of research is intermediate and society is patient. The research cycle that such a moonshot induces leads to researchers addressing more novel questions and producing more output.


\paragraph{Related Literature.}

The premise of our model is rooted in a large literature originating in the philosophy of science that studies how scientific research evolves; \citet{kuhn2012structure} is the most prominent contribution therein. \citet{kuhn2012structure} distinguishes between two phases in the evolution of knowledge: normal science and scientific revolutions. While his focus and, therefore, his hypotheses differ from ours, our findings share some similarities with his description of the dynamics of science. Our main distinction to the literature building on \citet{kuhn2012structure} is that we study how past discoveries shape today's choice of research questions and the search for answers rather than how an overreaching theory is used to tackle a given question. Building on \citet{kuhn2012structure}, \citet{kitcher1990division} introduces a model of researchers choosing between different theories. Further developing that idea, \citet{brock1999formal} show that social factors, such as the desire for conformity, affect whether a superior theory will outcompete an inferior theory over time. \citet{strevens2003role} argues that the priority reward system, as described in \citet{merton1957priorities}, can achieve an efficient allocation of resources across approaches. \citet{Akerlof2018} provide a dynamic economic model of competing paradigms and show that a false paradigm may survive due to homophily across generations of researchers. Relatedly, philosophers of science have studied how researchers uncover the most significant approach to a research topic \citep{alexander2015epistemic,weisberg2009epistemic}. They formalize approaches as locations in an epistemic landscape and their significance as the unknown height at that location. None of these ideas is our focus, yet our model connects naturally and provides a complementary perspective on cumulative research. 

Closer to an economic approach is a growing literature that leverages large-scale data sets to shed light on researchers' decision-making \citep[see][for an overview]{Fortunato2018}. \citet{foster2015tradition} empirically analyze researchers' strategies in biomedical chemistry and find that conservative research strategies, akin to little novelty and deepening knowledge in our setting, are more widespread than risky strategies, akin to high novelty and expanding knowledge in our setting. They emphasize a tension between being productive and making novel contributions. Building on these ideas, \citet{Rzhetsky2015} simulate how risk-taking speeds up the evolution of knowledge. \citet{myers2020elasticity} discusses the effectiveness of policy instruments to encourage more novelty in research. We view our work as complementary to this literature. We provide a theoretical model that, in its baseline static version, generates outcomes in line with their empirical findings. The tractability of our model allows us to generate further insights from making it dynamic and by adding policy tools to the picture.

An adjacent literature on innovation in corporate R\&D studies the allocation of resources to different innovative strategies. Our key departure from this literature lies in the value of the findings. In corporate R\&D, discovery is rewarded by a product's market value. In science, the value of a finding is the knowledge it generates, for example, to guide decision-making. \replaced{In the innovation literature, the direct spillovers of findings on the value of related questions are rarely the main driver of incentives. However, in our model the benefits of a finding specifically derive from improved knowledge about related questions.}{Therefore, in the innovation literature, the primary concern is with \emph{what} is found, whereas, in our model, the primary concern is with \emph{whether} something is found and how informative a finding is to decision-making.} \citet{hopenhayn2021direction}, for example, study the competitive dynamic allocation of researchers to different questions\replaced{, but assume the value of projects is independent of the stock of completed projects.}{. Apart from the different assumptions on the value of findings, we assume that discovering the answer to a question is informative about related questions, a feature absent in their model.} \citet{bryan2017direction} analyze a similar problem. In their setting, innovations may relate, but, in contrast to our setting, the sequence in which innovations can be achieved is exogenously fixed. 

A key aspect of our model is that the benefits and costs of addressing a research question depend on existing knowledge. The theoretical literature on scientific discoveries does not explicitly model this aspect, yet it incidentally captures parts of the scientific process we have in mind. \citet{aghion2008academic} consider a setting in which they assume that knowledge evolves in an exogenous step-by-step structure, whereas \citet{BRAMOULLE2010cycles} provide a model of research cycles albeit without considering an intensive margin. In our framework, intensive margins are relevant. 
    
Conceptually, we contribute to the literature modeling search as discovery on a Brownian path that builds on \citet{callander2011searching}. Our focus on modeling scientific research leads us to depart from the canonical ideas in the existing literature---most notably in two dimensions. 
 
First, our decision-maker aims to understand the entire Brownian path, as any question can become a potential problem to act on. In contrast, in \citet{callander2011searching,10.1093/restud/rdw008} knowing the location of the optimal realization along the path suffices for their decision-maker. Closer to us are \citet{bardhi2019attributes} and \citet{callander2017precedent}. Still, in \citet{bardhi2019attributes}, being informed about a summary statistic of the Brownian path suffices to make an optimal decision. In \citet{callander2017precedent}, being informed about the roots of the Brownian path suffices to make an optimal decision. In all of their models, the discovery of realizations beyond the frontier declines in value over time. \citet{callander2022innovation} is an exception where market competition slows this decline. Nevertheless, knowledge expansion eventually halts. In our model, there is a constant and endogenous desire to expand knowledge.\footnote{\citet{jovanovic1990long} study a related problem. There, expanding knowledge implies an i.i.d. draw at a fixed cost, while deepening knowledge is costless. Here, all questions are connected. See \citet{callander2014preemptive,callander2021power,Minipubs,callander2019risk,urgun2023contiguous} for unrelated applications in a related framework.} 

Second, we conceptualize discovering the realization of the Brownian motion at a particular point as a costly search process that may fail. This generates an endogenous link between novelty and output, leading to a trade-off: should the researcher choose more novel questions or higher research output? The existing literature ignores this link between novelty and output.

\added{Combining both differences to the literature, a new reason why learning stops emerges: the stochastic process taking an unexpected turn. Such turns prevent researchers from discovering the resulting unexpected answers, as their search optimally focuses on the most likely answers. This reason differs from those present in the literature where search tends to stop because further search becomes less valuable; e.g., due to very informative or valuable discoveries.}

Finally, the endogenous growth literature is related in that research and innovation generate value for society. Typically, the value of successful research derives from improvements in the product market, which is usually modeled in one of two ways: by expanding the set of available varieties (e.g., \citet{romer1990endogenous}) or by climbing a quality ladder, that is, by replacing old products or processes with improved ones (e.g., \citet{grossman1991quality,d2342669-e24e-3efd-8828-6d4269cc25ac}). In our model, a ladder structure of knowledge expansion somewhat reminiscent of the quality ladder model as described in \citet[][]{klette2004innovating} arises as well. However, there is a crucial difference. Knowledge advances horizontally in step sizes, but new steps leave the value of old ones unaffected. In that sense, the ladder structure arising in our model is closer to the expansion of product varieties as in \citet{romer1990endogenous}. Yet, we show that the ladder structure may be suboptimal.\footnote{It is important to keep in mind that questions in our model are related and have informational spillovers. This relation is absent in models with product varieties.}

\section{A Model of Knowledge and Research} 
\label{sec:model}

There are two players: society---represented by a single decision-maker---and a researcher. The researcher observes initial knowledge $\mathcal F_k$ and chooses two actions: a question, $x$, and a probability, $\rho$, with which she discovers the answer, $y(x)$, to question $x$. With probability $\rho$, knowledge is augmented by the question-answer pair and becomes $\mathcal F_k \cup \{(x,y(x))\}$. With complementary probability, knowledge remains $\mathcal F_k$. Finally, the decision-maker observes current knowledge and decides for every question whether to apply knowledge or to select an outside option.

\subsection{Knowledge and Conjectures}

\paragraph{Questions and answers.} We represent the universe of questions by the real line. A \emph{question} is an element $x \in \mathbb{R}$. Each question $x$ has a unique answer, $y(x) \in \mathbb{R}$. 

\paragraph{Truth and knowledge.} The answer $y(x)$ to question $x$ is determined by the truth. The truth is the graph of the realization of a random variable $Y(x)$ that follows a standard Brownian motion defined over the entire real line.\footnote{As in \citet{callander2011searching}, the realized truth $Y$ is a random draw from the space of all possible paths $\mathcal{Y}$ generated by a standard Brownian motion going through an initial knowledge point $(x_0,y(x_0))$.} This assumption captures the following notion: the answer to question $x$ is likely to be similar to the answer to a close-by question $x'$. As the distance between $x$ and $x'$ increases, the uncertainty increases. However, a correlation remains.

\emph{Knowledge} is the finite collection of known question-answer pairs. We denote it by $\mathcal F_k=\{(x_i,y(x_i))\}_{i=1}^{k}$. For notational convenience, we assume that $\mathcal F_k$ is ordered such that $x_i<x_{i+1}$. We refer to $x_1$ and $x_k$ as the \emph{frontiers} of knowledge. Knowledge $\mathcal F_k$ determines a partition of the real line consisting of $k+1$ elements \[\mathcal X_k :=\{(-\infty,x_1), [x_1, x_2),\cdots, [x_{k-1},x_k),[x_k,\infty)\}.\] 

Each element of the partition $\mathcal{X}_k$ is an \emph{area}. We call $(-\infty,x_1)$ area $0$, $[x_1,x_2)$ area $1$, and so on until area $k$, which is $[x_k,\infty)$. The length of area $i \in \{1,..,k-1\}$ is $X_i:=x_{i+1}-x_i$, and $X_0=X_k=\infty$.

\paragraph{Conjectures.} A conjecture is the cumulative distribution function $G_x(y|\mathcal F_k)$ of the answer $y(x)$ to question $x$ given knowledge $\mathcal F_k$. Because $Y(x)$ is determined by a Brownian motion, the conjecture about $y(x)$ 
is a cumulative distribution function of a normal distribution with mean $\mu_x(Y|\mathcal F_k)$ and variance $\sigma^2_x(Y|\mathcal F_k)$. Both $\mu_{x}$ and $\sigma^2_{x}$ follow from the properties of the Brownian motion. 

\begin{property}[Expected Value]\label{prpty:mu}
Given $\mathcal F_k$, the conjecture $G_{x}(y|\mathcal F_k)$ has mean:
\[\mu_x(Y|\mathcal F_k) = \begin{cases}
y(x_1) & \text{if}~ x <x_1 \\
y(x_{i}) + \frac{x-x_{i}}{X_i}(y(x_{i+1})-y(x_{i})) &\text{if}~ x \in [x_i,x_{i+1}), i\in\{1,...,k-1\}  \\
y(x_k) & \text{if}~ x \geq x_k.\end{cases} \]
\end{property}

\begin{property}[Variance]\label{prpty:variance}
Given $\mathcal F_k$, the conjecture $G_{x}(y|\mathcal F_k)$ has variance:
\[
\sigma^2_x(Y|\mathcal F_k) = \begin{cases}
	x_1-x & \text{if}~ x <x_1 \\
	\frac{(x_{i+1} - x)(x-x_{i})}{X_i} &\text{if}~ x \in [x_i,x_{i+1}), i\in\{1,...,k-1\}\\
	x-x_k & \text{if}~ x \geq x_k.
\end{cases}
\]
\end{property}

\subsection{Actions and Payoffs}
\paragraph{Decision-maker.} For each question $x$, the decision-maker either applies knowledge, or takes an outside option. The decision-maker's payoff from taking the outside option is normalized to zero. If she applies knowledge, her expected payoff is
  \[1-\frac{\sigma^2_x(Y\mid \mathcal{F}_k)}{q},\]
with $q>0$ exogenously given.\footnote{One microfoundation to attain these payoffs is to assume that the decision-maker's preferences are represented by an affine transformation of a quadratic loss function and she aims to match the truth. She would then choose $\mu(x)$ as her best reply. Normalizing the outside option to zero, choosing an intercept of one and a slope of $1/q$ provides our payoff formulation.} Abstracting from any prioritization of questions, the decision-maker values all questions equally. Her value of knowing $\mathcal{F}_k$ is 
  \[v(\mathcal F_k):=\int_{-\infty}^\infty \max\left\{1 - \frac{\sigma_x^2(Y|\mathcal F_k)}{q},0\right\} \mathrm{d}x.\]

\subparagraph{The decision-maker's outside option.} A key element of our model is that the precision of knowledge determines whether the decision-maker applies it to address problems or prefers the outside option, $\varnothing$. This feature captures what is referred to as the ``precautionary principle'': if uncertainty is large, prudence trumps risking poor application of knowledge. How often the decision-maker reverts to the outside option is governed by the parameter $q$, the size of which is not essential to our results as long as $q \in (0,\infty)$. Yet, $q$ gives knowledge creation a meaning. To see this, consider the two limiting cases. As $q \rightarrow 0$, the decision-maker does not apply knowledge unless she is certain about the answer. As knowledge is finite, but problems are uncountably infinite, the decision-maker chooses the outside option almost everywhere for any knowledge: knowledge becomes irrelevant to decision-making. At the other extreme, $q \rightarrow \infty$, the decision-maker prefers to apply knowledge no matter how vague the conjecture. In that case, the emphasis the decision-maker puts on the precision of her answers must be low: knowledge becomes irrelevant to decision-making.

\paragraph{Researcher.} The researcher builds on initial knowledge $\mathcal F_k$ and decides to search for an answer to a question $x$. Another key element of our model is that the researcher decides how much effort to exert to find the answer. Given a choice of question $x$, we posit a one-to-one relationship between the level of effort exerted and the resulting probability $\rho$ of discovering the answer $y(x)$. Thus, $(x,\rho)$ is a sufficient statistic to summarize the researcher's choice.  

To save on notation, we allow the researcher to choose $\rho$ directly. We aim to capture that increasing the probability of discovery requires costly effort. In \cref{sub:selecting_a_research_questions}, we provide details and a microfoundation. For now, we assume the researcher's cost to be
\[\hat{c}(\rho;x) := \eta \big(\operatorname{erf}^{-1}(\rho)\big)^2 \sigma^2_x\big(Y|\mathcal F_k\big),\]
where $\eta\geq 0$ is an exogenous cost parameter and $\operatorname{erf}^{-1}(\cdot)$ is the inverse error function of the normal distribution. The cost scales in uncertainty: discovering $y(x)$ with a given probability $\rho$ is more costly if the conjecture about $y(x)$ is less precise.

On the benefits side, we assume that the researcher is aligned with the decision-maker. The researcher's total payoff is 
\[\rho\bigg(v\Big(\mathcal F_k \cup \{(x,y(x))\}\Big)-v(\mathcal F_k)\bigg)  - \eta \Big(\operatorname{erf}^{-1}(\rho)\Big)^2 \sigma^2_x\big(Y|\mathcal F_k\big).\]
Economically, our assumption entails that the researcher benefits from her discovery through its impact on the decision-maker's payoff. We revisit this assumption and its implications when we get to the researcher's problem (\cref{sub:selecting_a_research_questions}) and again in a dynamic version of the model in \cref{sub:research_dynamic}. 

\section{The Benefits of Discovery} 
\label{sec:the_value_of_increasing_knowledge}

\added{In this section, we describe the benefits of a particular discovery. That is, we ignore how the discovery came about and focus on the marginal value it creates.}

Discovery occurs whenever an answer is found and the new question-answer pair is added to existing knowledge, $\mathcal F_k$. 
The (gross) benefits of discovering $y(x)$ are the marginal increase in the value of knowledge from adding $(x,y(x))$ defined as
\[V(x;\mathcal{F}_k):=v(\mathcal F_{k} \cup \{(x,y(x))\})-v(\mathcal F_{k}).\]
We distinguish two scenarios: expanding knowledge and deepening knowledge. A discovery $y(x)$ \emph{expands} knowledge if $x \notin [x_1,x_k]$. A discovery $y(x)$ \emph{deepens} knowledge in area $i$ if $x \in [x_i,x_{i+1}]$.

The benefits of a discovery are determined by the length of the research area, $X$, the discovery occurs in, and how distant the question, $x$, is from existing knowledge.

\begin{definition}[Distance]
The distance of question $x$ from knowledge $\mathcal F_k$ is the minimal Euclidean distance to a question to which the answer is known:
 \[d(x):=\min\limits_{\xi \in \{x_1,x_2,...x_k\}} |x-\xi|.\]
\end{definition}
\begin{definition}[Variance]\label{def:variance}
 The variance of a question with distance $d$ in an area of length $X$ is \(\sigma^2(d;X):=d(X-d)/X.\)
\end{definition} 

Note that $\sigma^2(d;X)=\sigma^2_x(Y|\mathcal F_k)$ for $d(x)=d$ and $x$ in an area of length $X$. Abusing notation, we define \added{for any $X \in \mathbb{R}\cup \{\infty\}$ and $d \in [0,X/2]$}
\[\begin{split}V(d;X):= 
 \frac{1}{6q} \Big(2X \sigma^2(d;X) &+ \boldsymbol{1}_{d>4q} \sqrt{d}(d-4q)^{3/2} \\  &+\boldsymbol{1}_{X-d>4q} \sqrt{X-d} ~(X-d-4q)^{3/2}\\&- \boldsymbol{1}_{X>4q} \sqrt{X}(X-4q)^{3/2}  \Big),\end{split}\]
\replaced{where $V(d;\infty){:=} \lim_{X \rightarrow \infty} V(d;X)$.}{Moreover, let $V(d;\infty){:=} \lim_{X \rightarrow \infty} V(d;X)$.}

\begin{proposition}\label{prop:value_knowledge}
	$V(d;X)$ describes the benefits of a discovery $(x,y(x))$ with distance $d(x)=d$ to existing knowledge when the question $x$ lies in an area of length $X$.
\end{proposition}

\cref{prop:value_knowledge} shows that the benefits of a discovery depend only on the parameters $d$ and $X$ of a question, not its precise location. 
To gain some intuition on $V(d;X)$, note that the terms without an indicator function measure the direct reduction in uncertainty about answers due to the discovery. The indicator that enters negatively becomes active if the decision-maker chose the outside option for some question in the area \emph{before} discovery. The indicators that enter positively become active only if the decision-maker chooses the outside option for some question \emph{after} discovery. 
\deleted{The right panel of Figure 2
above illustrates the benefits-of-discovery function $V(d,X)$.} 
\begin{figure}\centering
	\subfloat{\includestandalone[width=.45\textwidth]{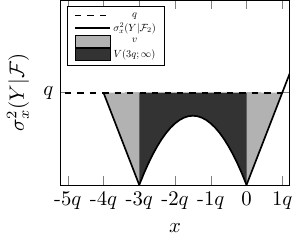}}\hfill
	\subfloat{\includestandalone[width=.45\textwidth]{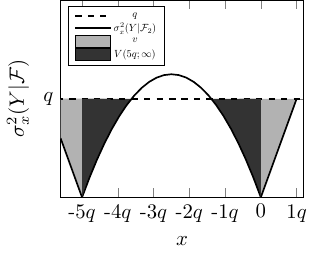}}
	\caption{\emph{Benefits-maximizing (left) and too large (right) distance of $x$ given $\mathcal F_1$.}}\label{fig:optimalandwrongchoice}
\end{figure}
\paragraph{Expanding Knowledge.} Suppose we expand knowledge by discovering a question-answer pair $(x,y(x))$ with $x<x_1$. Our discovery pushes the knowledge frontier to the left generating the area $[x,x_1)$. The benefits of this knowledge expansion come from the value of the new area $[x,x_1)$---the dark-shaded area in \cref{fig:optimalandwrongchoice}'s left panel.\footnote{More precisely, the conjectures about questions to the left of the old frontier are replaced by conjectures inside the new research area, and conjectures to the left of the new frontier also become more precise. As can be seen in the left panel of \cref{fig:variance2}, the variance reduction to the left of the frontier is always the same. Hence, the benefits are the same as if only the new area was added.}

The value of the new area depends on \replaced{the amount of questions with conjectures based on two discoveries and the degree of improvement in those conjectures relative to the outside option.}{ (i) the amount of questions it contains and (ii) the degree of improvement in decision-making relative to the outside option.}  The benefits-maximizing question resolves a marginal-inframarginal trade-off. Increasing the length of the newly created area has two opposing effects on the value: \replaced{(i) a positive marginal effect because the amount of questions in the new area increases, and (ii) a negative inframarginal because the degree of improvement decreases.}{ The amount of questions with improved conjectures increases. However, the increase in area length decreases the precision of conjectures about inframarginal questions in it.}

\begin{figure}[t]
	\centering
	\subfloat{
		\includestandalone[width=.45\textwidth]{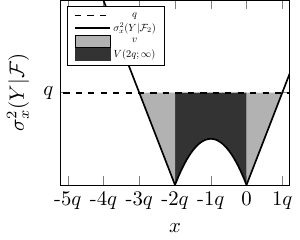}}
	\subfloat{
	\includestandalone[width=.45\textwidth]{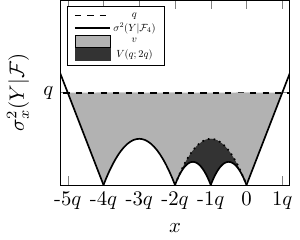}}
	\caption{\emph{The benefits of discovery.} \footnotesize{The dark-shaded area illustrates the benefits of discovery; the light-shaded areas illustrate the value of initial knowledge. In the \emph{left panel}, knowledge is expanded. In the \emph{right panel}, knowledge is deepened.}}\label{fig:variance2}
\end{figure}
\Cref{fig:optimalandwrongchoice} and \Cref{fig:variance2} illustrate the benefits of discovery from creating ideal (left panel of \cref{fig:optimalandwrongchoice}), too large (right panel of \cref{fig:optimalandwrongchoice}), and too short (left panel of \cref{fig:variance2}) areas. The largest benefits come at an intermediate level at which all conjectures have a variance strictly smaller than $q$.

\paragraph{Deepening knowledge.} Deepening knowledge differs conceptually. Instead of creating a new area, discoveries replace an existing area with two shorter ones. Discovering the answer to $y(x)$ replaces an area, $[x_i,x_{i+1})$ with the new areas $[x_i,x)$ and $[x,x_{i+1})$.

Areas of length $3q$ provide the largest benefits, the one depicted in the left panel of \cref{fig:optimalandwrongchoice}. Thus, the benefits-maximizing discovery inside an area of length $X_i=6q$ is naturally at the midpoint, replacing the area by two areas of length $3q$.

Finding the benefits-maximizing discovery for areas of length $X_i\neq 6q$ is less straightforward. There are two forces at play. First, there is a benefit to replacing the old area with two symmetric new areas. The intuition echoes that of expanding knowledge: the inframarginal loss increases when an area becomes too large. Choosing symmetric area lengths reduces the inframarginal losses compared with asymmetric area lengths. Second, benefits decline if the area length is greater than $3q$ because conjectures inside the area become imprecise. 

If the initial area length $X_i$ was small, the first force would dominate. Selecting the midpoint at $d=X_i/2$ is optimal. However, if $X_i$ was large, the trade-off would be resolved in favor of creating one high-value area at the cost of having imprecise conjectures in the other. A cutoff, $\widetilde{X}^0$, determines which force dominates.

\paragraph{Expanding vs. Deepening Knowledge.} Finally, there is a trade-off between expanding knowledge and deepening knowledge. On the one hand, expanding knowledge means that no area is replaced and a new area is created. On the other hand, deepening knowledge means creating two areas with precise conjectures. Expanding knowledge provides higher benefits than deepening knowledge in some existing area only if all existing areas are shorter than a cutoff $\widehat{X}^0$. \cref{fig:benefit} illustrates this observation.

\begin{figure}[bt]
	\subfloat{
	\includestandalone[width=.45\textwidth]{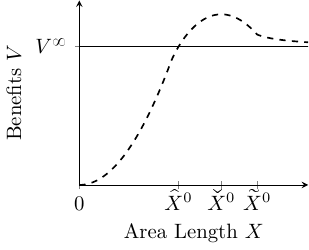}}\hfill
	\subfloat{\includestandalone[width=.45\textwidth]{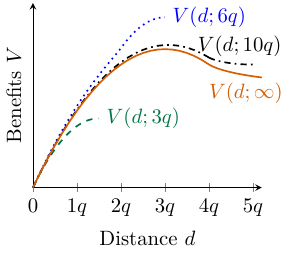}}
	\caption{\emph{The benefits of discovery.} \footnotesize{The dashed line in the \emph{left panel} plots the benefits of discovery $V(d^0(X);X)$ against $X<\infty$. The solid line shows the maximum benefits from expanding knowledge. The \emph{right panel} plots $V(d;X)$ for different $X$.
	}\label{fig:benefit}}
	\end{figure}

We now summarize our findings\deleted{ in a corollary to \mbox{\cref{prop:value_knowledge}}}. We define the optimal distance as \[d^0(X):= \arg \max_d V(d;X).\]

\begin{proposition}\label{cor:deepen}\label{cor:opt}\label{cor:maxinteriorsmaller8}
	A distance of $d^0(\infty){=}3q$ maximizes the benefits of expanding knowledge. 
	When deepening knowledge in an area with length below a cutoff $\widetilde{X}^0$, the midpoint of the area maximizes the benefits, $d^0(X) = X/2$. Otherwise, a distance between $3q$ and the midpoint maximizes the benefits of deepening knowledge, $d^0(X)\in (3q,X/2)$.

	Expanding knowledge is only benefits-maximizing if all available \added{bounded} areas are shorter than a cutoff, $\widehat{X}^0$, with $\widehat{X}^0<\widetilde{X}^0$. The benefits from optimally deepening knowledge are single peaked in the area length with the peak at $\widecheck{X}^0\in(\widehat{X}^0,\widetilde{X}^0)$.
\end{proposition}

\section{The Researcher} 
\label{sub:selecting_a_research_questions}
\subsection{The Researcher's Objective} 
In this section, we analyze the researcher's optimal choice. Recall that the researcher only benefits from research if it culminates in a discovery.\footnote{This is a direct consequence of the assumption that the decision-maker only has access to $\mathcal F_k$ when addressing problems. One rationale is a moral-hazard concern: science is complex, and it is impossible to distinguish the absence of a finding from the absence of a (proper) search.} The researcher's expected payoff given a choice of question $x$ and success probability $\rho$ can be written as
\[u_R(d,\rho;X):= \rho V(d;X)- \eta \underbrace{\tilde{c}(\rho)\sigma^2(d;X)}_{=c(\rho,d;X)},\]
where  $\tilde{c}(\rho):=(\operatorname{erf}^{-1}(\rho))^2$. To obtain this expression, we replaced the question $x$ by its sufficient statistics $(d, X)$.

The cost of research, $\eta c(\rho,d;X)$, derives from conceptualizing research as the search for an answer. We assume that, given a question $x$, the researcher chooses a sampling interval $[a,b] \in \mathbb{R}$ in the $y$-dimension. She discovers the answer if and only if $y(x) \in [a,b]$. We interpret the interval length as the amount of effort the researcher invests in finding an answer. For simplicity, we assume a quadratic cost $ \eta (a{-}b)^2$.
 
\begin{lemma}\label{lem:cost}
For knowledge $\mathcal{F}_k$, probability $\rho$, and question $x$, the minimal cost of obtaining an answer to question $x$ with probability $\rho$ is
\[\eta c(\rho,d;X)= \eta \tilde{c}(\rho) \sigma^2(d;X).\]
\end{lemma}

\replaced{The intuition behind \Cref{lem:cost} is straightforward. Conjectures are normally distributed and it is optimal to center the search effort around the mean. Given a success probability $\rho$, the more precise the conjecture, the shorter the expected search around the mean, and thus, the lower the cost. For a given variance, the higher the desired success probability, the larger the respective sampling interval. }{ In Online Appendix E, we show that any (i) homogeneous, (ii) increasing, and (iii) convex sampling cost function over $b-a$ implies a reduced-form cost function similar to the one we impose. Therefore, such an alternative cost function would not alter our results qualitatively. A general increasing and convex $\tilde{c}(\rho)$ is possible, yet harder to microfound. We discuss the desirable properties of $\tilde{c}(\rho)$ in Online Appendix F.}

Our cost function exhibits the following properties: It is (i) multiplicatively separable in $\rho$ and $(d;X)$, (ii) increasing in $d$ and $X$, and (iii) concave in $d$; the concavity decreases in $X$ and becomes linear in $d$ in the limit as $X \rightarrow \infty$.\footnote{\added{More complex cost functions are possible, we discuss two alternatives in Online \cref{sec:the_cost_of_research} and \ref{sec:Notation}.}} 

The cost of research links output and novelty. To see this, consider a researcher who chooses a question $x$ and aims to discover its answer with probability $\rho$. If that researcher increases $\rho$ by a given amount, $\varepsilon$, she needs to increase her search effort, that is, expand her sampling interval, $[a,b]$\added{, which is optimally centered around the mean}. The additional effort required to increase $\rho$ depends on the variance of the conjecture about $y(x)$, and hence, on the novelty of the question. If the variance is low, the density at the boundaries of the sampling interval is high---a moderate increase in effort gets the success probability to $\rho+\varepsilon$. If, instead, the variance is high, the density is low and the increase in effort must be larger.

\subsection{The Researcher's Choice} 
We now characterize the researcher's optimal choice and elaborate on the resolution of the novelty-output trade-off. The researcher solves
\[\max_{X \in \{X_0,...,X_k\}}\quad \underbrace{\max_{\substack{d \in [0,X/2],\\\rho \in [0,1]}} \rho V(d;X)- \eta c(\rho,d;X)}_{=:U_R(X)}.\]
Without cost ($\eta=0$), we can apply \cref{prop:value_knowledge} to derive the researcher's optimal choice. For any research area of length $X$, the researcher selects distance $d^0(X)$ and discovers an answer with certainty. 

However, for positive cost, $\eta>0$, the researcher's optimal decision on output, $\rho^\eta(X)$, is nontrivial and linked with her decision on novelty, $d^\eta(X)$. Choosing a question close to existing knowledge allows for a high probability of discovery at a low cost. The researcher's initial conjecture about the answer is already precise. Nevertheless, her payoff is low, as such a discovery provides little benefits. By increasing the distance, the researcher increases both benefits and cost, ceteris paribus. The effect on the optimal probability of discovery is ambiguous. Depending on which effect dominates, the distance and the probability of discovery are substitutes (the researcher optimally reduces the success probability when answering a more novel question) or complements (she increases the success probability when answering a more novel question).

\paragraph{Optimal choice within a research area.}

The following proposition captures the key aspects of the researcher's choice within a research area. \deleted{Figure X illustrates it.}

\begin{proposition}\label{prop:substitutes}
Suppose $\eta>0$. Researchers fail with positive probability, $\rho^\eta(X)\in (0,1)$. There is a cutoff area length $\widetilde{X}^\eta<\widetilde{X}^0$ such that researchers choose the benefits-maximizing distance, $d^\eta(X)=d^0(X)$, in area $X$ if and only if $X\leq \widetilde{X}^\eta$. Otherwise, researchers choose a question strictly less novel than that, $d^\eta(X)<d^0(X)$.
\end{proposition}

\Cref{prop:substitutes} shows that when expanding knowledge the researcher chooses a question closer to existing knowledge than the benefits-maximizing distance $3q$. This is because novelty and output are substitutes. The marginal cost of increasing $\rho$ rises with $d$, while the marginal benefits of increasing $d$ approach zero as $d \rightarrow 3q$. The researcher balances novelty and output and selects a question less novel than $3q$.

That trade-off is less pronounced when the researcher deepens knowledge. The reason is that inside an area moving away from one boundary implies moving closer to the other boundary. Thus, the marginal cost of the success probability flattens in distance and becomes zero at $d=X/2$. Whether this effect is strong enough to make novelty and output complements depends on area length.

\paragraph{Optimal choice among areas.} 
The following proposition characterizes the researcher's optimal choice among intervals\replaced{.}{and is illustrated in the right panel of Figure X.}


\begin{proposition}\label{prop:researchX}
Suppose $\eta>0$. There are cutoffs $2q<\widehat{X}^\eta<\dot{X}<\widecheck{X}^\eta<\widetilde{X}^\eta<8q$ such that the following claims hold: \begin{enumerate}
    \item\label{item:hat} The researcher expands knowledge if and only if knowledge is dense, that is, if and only if all bounded areas are shorter than $\widehat{X}^\eta \added{<\widehat{X}^0.}$
    \item\label{item:dot} Both novelty and output are nonmonotone in area length. Novelty attains a maximum at $\widetilde{X}^\eta$. Output attains a maximum at $\dot{X}$.
    \item\label{item:check} The researcher's expected payoff from conducting research in an area $X$, $U_R(X)$, is single peaked and attains a maximum at $\widecheck{X}^\eta$.
\end{enumerate}
\end{proposition}

\Cref{prop:researchX} shows that the pattern in distance is qualitatively the same as in \Cref{cor:maxinteriorsmaller8}. However, the cost adds another dimension: the researcher's choice of success probability interacts with both the choice of distance and research area.

Consider a short area. The scope of improving the decision-maker's policies is small because conjectures are already precise and investing in discovery thus provides a limited payoff. Consequently, the researcher does not invest much in the search for an answer despite the low cost. She opts for a low success probability. A marginal increase in the area length provides larger increase in the benefits than in the cost. In response, the researcher increases both distance and success probability.  

By contrast, consider a large area. The benefits of discovery trump those of the small area, yet the cost is higher. The researcher does not invest much in discovery due to the high marginal cost of increasing the success probability. As a result, the probability of discovery is low. If, in this case, the area length increases marginally, the researcher responds by decreasing both the distance and the success probability.

Finally, consider an area of intermediate length. The benefits of discovery are high, yet the associated cost is limited. The return on investment is large, and the probability of discovery is high. As the area length further increases, the marginal return of increasing the distance declines, while the marginal cost rises. Eventually, the researcher faces a trade-off: should she reduce the success probability to maintain maximal distance? It turns out that she should. While the researcher wants to remain at a boundary in her choice of distance, she mitigates the increased cost by lowering the success probability. 

The researcher's preferred area length, $\widecheck{X}^\eta$, is in a region in which a trade-off between output and novelty exists. While the researcher would prefer a larger research area to increase the benefits of research, she would prefer a smaller research area to reduce her cost. Thus, distance is increasing and the success probability is decreasing at the point at which the researcher's payoff is maximal.

Thus far, we have not taken into account which research areas are available, which is determined by existing knowledge $\mathcal F_k$. Computing the optimal among the available areas is straightforward. 

\section{The Evolution of Knowledge, Moonshots, and Research Cycles} 
\label{sub:research_dynamic}

In this section, we consider a dynamic extension of our baseline model to study the endogenous evolution of knowledge. Through the lens of our dynamic model, we find that forward-looking interventions, which induce research cycles through moonshot discoveries, can improve the evolution of knowledge. 

\subsection{Sequential Research} 
\label{sub:the_evolution_of_knowledge}
Our starting point is a setting in which knowledge is $\mathcal F_0=\{(x_0,y(x_0))\}$. At any time $t=1,2,\dots,$ a short-lived researcher $R_t$ arrives, observes current knowledge, $\mathcal{F}_{t-1}$, and selects $(x,\rho)$. If a discovery occurs, knowledge updates and a decision-maker updates the application of knowledge accordingly.

To retain focus, we make three assumptions: All researchers have the same cost parameter $\eta>0$,  break ties identically, and condition their decision $(x,\rho)$ \emph{only} on current knowledge. While we impose the first two assumptions for simplicity only, the last assumption is more meaningful.\footnote{In particular, once a researcher fails, all subsequent ones ignore those failures because they condition their decision only on current knowledge, and---by symmetry---fail again.} We discuss our modeling choice in \cref{sec:final_remarks}. 

Moving forward, keep in mind that, given our symmetry assumptions, we can invoke \cref{prop:researchX} to describe knowledge as \emph{dense} whenever $\mathcal F_k$ is such that all bounded areas are shorter than $\widehat{X}^\eta$. For any dense knowledge $\mathcal F_k$, each researcher's incentives are identical and in particular identical to the \replaced{incentives of a researcher facing the initial knowledge $\mathcal F_0$. }{ simplest dense knowledge $\mathcal F_0$.} As a benchmark, we first show that knowledge evolves in knowledge-expanding steps.

\begin{proposition}\label{prop:laissez-fair}
For any $\eta \geq 0$, knowledge is dense in any period. Specifically, every researcher aims to expand knowledge by the same distance $d^\eta(\infty) \in (2q,3q]$ with the same probability of discovery $\rho^\eta(\infty) \leq 1$. \added{Moreover, both the distance $d^\eta(\infty)$ and the probability of discovery $\rho^\eta(\infty)$ strictly decrease in the cost parameter $\eta$.}
\end{proposition}

\Cref{prop:laissez-fair} shows that no short-lived researcher endogenously inspires future researchers to deepen knowledge. The intuition follows from \cref{prop:value_knowledge,prop:researchX}: Researchers are rewarded for their immediate contribution to the value of knowledge. Therefore, no researcher expands knowledge beyond the benefits-maximizing $d^0(\infty)=3q$. Deepening knowledge within areas of $X\leq 3q$, however, is unattractive as it creates too little value. Thus, all researchers aim to expand knowledge, and the stepsize depends on $\eta$. The larger $\eta$, the smaller the stepsize $d^\eta(\infty)$. Moreover, if $\eta>0$, researchers fail at each step with probability $1-\rho^\eta(\infty)>0$ \replaced{ if the Brownian takes an unlikely turn. If one researcher fails, all later researchers fail too. By symmetry, they make the exact same choice as the initially failing researcher and thus fail to discover the realization. }{ and such failure is permanent as all future researchers will make the same choices.}

\subsection{Interventions}\label{sub:interventions}

In reality, we see that societies invest in affecting scientists' choices. The most prominent incentive schemes are large ex-ante cost reductions through grants from funding institutions, such as the NSF or the ERC, and high-prestige ex-post rewards for successful research, such as the Nobel Prize or the Fields Medal. These incentives, however, are often awarded to scientists pushing the frontier considerably.

A natural question is whether a designer (a funder, a government, \ldots) has an incentive to interfere with the knowledge production process in our model. Here, we consider a designer whose per-period payoff is a weighted average of the decision-makers' and the researchers' payoffs,
\begin{align}\label{eq:designer}
    \mathbb{E}\left[\sum_{t=1}^\infty \delta^{t-1} \Big((1-\alpha)v(\mathcal{F}_{t+1}) + \alpha \left(v(\mathcal{F}_{t+1})-v(\mathcal{F}_t) - \eta \hat{c}(d,\rho) \right) \Big)\right],
\end{align}
where $\alpha \in [0,1]$ measures the designer's weight on the researchers' payoffs. This formulation introduces two potential frictions between the designer and the researcher. 

First, their per-period payoffs differ. The only meaningful difference between the researcher's and the decision-maker's payoff is that the researcher bears the cost of research.\footnote{Note that the term $\alpha v(\mathcal{F}_t)$ is, within a period, a constant as the knowledge at the beginning of period $t$, $\mathcal{F}_t$ is exogenously given.} Therefore, we can think of the designer as having a modified cost parameter $\hat{\eta}\in[0,\eta]$, where $\hat{\eta}=0$ corresponds to full alignment with the decision-maker, and $\hat{\eta}=\eta$ to full alignment with the researcher. We can interpret $\hat{\eta}$ either as a degree of research cost internalization or, alternatively, as a measure of the appropriability of the benefits of research. Second, the time horizons of the researcher and the designer differ. Researchers are rewarded for their discoveries only through the immediate benefits they generate in the value of knowledge. However, a researcher's discovery alters the landscape of knowledge indefinitely and thereby the future value of knowledge as well. We assume that the designer discounts the future value of knowledge by $\delta \in [0,1)$. 

For now, we assume that the designer can intervene at most once and proceed by studying a set of natural benchmarks to understand the forces at play. First, we isolate the time-horizon friction by assuming that research is costless ($\eta=0$) and that the designer is forward-looking ($\delta >0$). Second, we isolate the cost friction by considering positive research cost ($\eta>0$) but a myopic designer ($\delta=0$). 

\paragraph{Forward-looking designer and costless research.}
First, consider a forward-looking designer, $\delta \in (0,1)$, in an environment with costless research, $\eta=0$. Note that this designer's payoff is 
\begin{align}
    \mathbb{E} \left[\sum_{t=1}^\infty \delta^{t-1} \left( v(\mathcal{F}_{t+1}) - \alpha\cdot v(\mathcal{F}_{t}) \right)\right].
\end{align}

 \replaced{A researcher without cost chooses $d^0(\infty)=3q$ and $\rho^0(\infty)=1$ (\cref{cor:opt}). Moreover, by \Cref{prop:laissez-fair}, knowledge remains dense under these choices. From \cref{cor:opt}, we know that adding an area of length $3q$ to any given set of areas $\mathcal X_k$ dominates adding a single bounded area of any other length. Because every researcher adds such an area, it follow that knowledge expands in $3q$-steps and without failure. Irrespective of $(\alpha,\delta)$, the designer has no reason to intervene }{ Recall that a researcher without cost behaves like the decision-maker and chooses $d^0(\infty)=3q$ (Corollary 1) and $\rho^0(\infty)=1$ and that, by Proposition 4, knowledge is dense in all periods without intervention. It follows immediately that, for any $\alpha$, knowledge will evolve in step-by-step expansions without failures. Moreover, it evolves with the optimal stepsize for all players. The designer has no reason to intervene}, and a sequence of short-lived researchers without cost implements the forward-looking designer's optimum \deleted{for any $(\alpha,\delta)$}. Thus, it is not the researcher's short-livedness alone that causes inefficiency in the evolution of knowledge.

\paragraph{Myopic designer and costly research.} 
Second, consider a myopic designer, $\delta=0$, in an environment with costly research, $\eta>0$, that is not fully aligned with the researcher, $\alpha<1$. Such a designer has, in any period $t$, payoffs similar to this period's researcher, albeit with a smaller cost parameter, $\hat{\eta}$. It follows that, in each period, the researcher's choices are suboptimal from the designer's perspective. In particular, the designer prefers a larger distance and a higher success probability. 

Straightforwardly, such a myopic designer \replaced{wants to align}{can achieve his per-period optimum by aligning} the researcher's cost parameter, $\eta$, with his own, $\hat{\eta}$, whenever given the chance. \replaced{If the designer aligns the cost,}{ If the designer has the tools to align the costs,} the researcher implements the designer's preferred research. \added{Whether the designer aligns $\eta$ and $\hat{\eta}$ depends on whether she needs to pay an additional cost to do so.} \replaced{If not, the designer}{While the designer has an incentive to} intervene\added{s} \deleted{in this setting} by reducing the researcher's cost parameter directly, \added{yet} the evolution of knowledge remains qualitatively unchanged.\footnote{Because $\eta$ governs the relative weight of benefits and cost, increasing the researcher's benefits or decreasing her cost is the same.} Knowledge evolves through step-by-step expansions albeit with larger stepsizes and fewer failures. This observation also provides an answer to the question of how heterogeneous researchers would change the picture. In fact, heterogeneous researchers are researchers that incorporate the cost at different weights, thus observationally indistinguishable from cost-internalizing designers.

To achieve a better outcome, the designer could, for example, award a grant to the researcher which effectively lowers that researcher's cost by allowing her to hire personnel or reduce her teaching load. \added{However, such subsidies may come at a cost to the designer not reflected in \eqref{eq:designer}. If such a cost was present, the designer would subsidize the researcher only until the marginal cost of doing so exceeds the marginal benefits.}\footnote{\added{A realistic model would need to develop a theory of supply of grants, awards, or teaching reductions, which we expect to result in a non-trivial cost function to the designer. We thus leave it for future research.}}

\paragraph{Moonshots.} 

The preceding benchmarks provide a rationale for moderate interventions at most that encourage researchers to invest more into the search for answers to more distant questions. These observations highlight that it is neither the researcher's short-livedness nor her cost of research alone that warrants more substantial interference with the scientific production process. To rationalize more drastic interventions that qualitatively change the evolution of knowledge, we now turn to a forward-looking designer in an environment with research cost.

We focus on a particular simple model in which the designer does not internalize the researcher's cost ($\alpha = 0$). We allow the designer to fully control the first researcher's actions at no cost. We are interested in whether such a designer has an incentive to induce a \emph{moonshot}, a discovery more distant than the myopically optimal distance, $3q$, to the knowledge frontier. Because a distance of exactly $3q$ maximizes the value of the area generated (by \cref{cor:opt}), a moonshot can only be beneficial if it ``inspires'' future researchers to deepen knowledge in the newly created area; that is, if it is larger than $\widehat{X}^\eta$. The next proposition states that a designer finds it optimal to launch a moonshot if and only if he is sufficiently patient and the researcher's cost friction is intermediate.

\begin{proposition}\label{prop:moonshots}
    \added{Suppose $\alpha=0$. }There are cost parameters $0<\underline{\eta}<\overline{\eta}<\infty$ and a critical discount factor $\underline{\delta}(\eta)<1$ such that for $\eta \in (\underline{\eta},\overline{\eta})$ the designer optimally launches a moonshot if and only if $\delta>\underline{\delta}(\eta)$. If $\eta=0$ \replaced{or}{and} $\eta \rightarrow \infty$, a moonshot is suboptimal for any $\delta$.
\end{proposition}

\cref{prop:moonshots} incorporates the no-cost case. As we have argued above, moonshots are of no use in that case. The same is true if the researchers' costs are high. Then, $\rho^\eta$ is low, and already the second researcher most likely fails to discover an answer. Thus, the designer expects little progress from $t=2$ onward and focuses on maximizing the value of knowledge at the end of $t=1$ inducing the myopic optimum $d=3q$ \added{which, unless $\eta=0$, still implies more novelty than what the researcher chooses herself}.

\begin{figure}[t]
    \includestandalone[width=.46\textwidth]{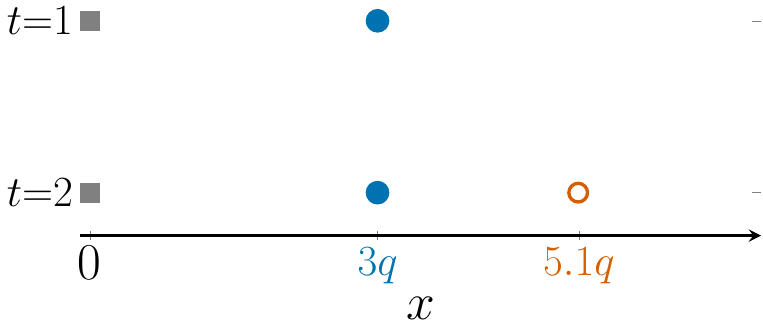}\hfill
    \includestandalone[width=.46\textwidth]{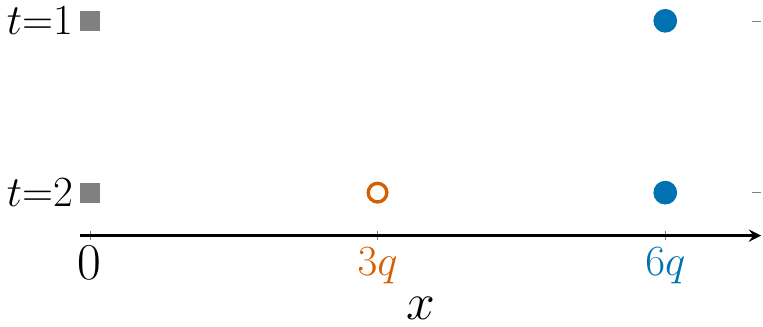}
    \caption{\emph{Evolution of knowledge for different choices in $t=1$.} \footnotesize{The dots show which questions have a known answer at each time $t$, assuming that discovery has been successful. The designer's choice for $R_1$ (\textcolor{Blue}{\large{$\bullet$}}) is given, $R_2$'s choice (\textcolor{Vermilion}{$\boldsymbol{\circ}$}) is a best response. $F_0={0,y(0)}$, $\eta=1$. The left panel assumes discovery of question $x=3q$ in period $t=1$, the right panel $x=6q$.}}\label{fig:evolution}
\end{figure}
    
For intermediate cost parameters, however, moonshots generate valuable spillovers. Without intervention, intermediate cost parameters lead to successful but too narrow research that fails too often. Researchers expand knowledge in a ladder-type structure but select questions too close to existing knowledge and too low success probabilities. In anticipation, the designer initiates a \emph{research cycle} through a moonshot. The next researcher builds on both the initial knowledge $\mathcal{F}_0$ and the moonshot discovery and deepens knowledge in the newly generated area in an attempt to close the gap between moonshot and initial knowledge.
    
A completed research cycle may lead to a more valuable landscape as \cref{fig:evolution} illustrates. Once the cycle is complete and knowledge is dense again, knowledge is in a better state than with one optimal area of length $3q$ and subsequent knowledge-expanding steps following \Cref{prop:laissez-fair}. As time moves on, the better landscape remains and generates further rents. 

\cref{fig:hazard_rate} illustrates a second positive effect of a moonshot. The probability that the evolution of knowledge ends because the next researcher fails is smaller after the moonshot than after a myopic disclosure. The moonshot opens a research area of considerable size. In $t{=}2$, the second researcher aims to fill it, but can use both the initial frontier and the moonshot discovery to fine-tune her research. This reduces her cost---in particular when the moonshot is not too far from initial knowledge. As a consequence of this logic, ``marsshots'' very far from existing knowledge are not optimal. They are too disconnected to inspire the following researchers to be productive. 

    \begin{figure}
        \centering
                \begin{tikzpicture}
                \begin{axis}[%
                    title={Hazard Rate of Science},
                    width=0.5\textwidth,
                    height=0.3\textwidth,
                    axis line style= ultra thick,
                    at={(0\textwidth,0\textwidth)},
                    scale only axis,
                    xmin=-0,
                    xmax=1.2,
                    ymin=0,
                    ymax=0.25,
                    xtick={0,1,2,3,4,5},
                    xticklabels={1,2,3,4,5,6},
                    ytick={0, 0.2,0.5,0.75,1},
                    xlabel={Time},
                    ylabel={P(science has stopped)},
                    axis x line=bottom,
                    axis background/.style={fill=white},
                    axis y line=left
                    ]
                    \node (start) at (0,0) [circle,fill,inner sep=1.5pt]{};
                    \node (t1m) at (1,1-0.8355) [circle,fill,inner sep=1.5pt]{};
                    \node (t2m) at (2,1-0.6981) [circle,fill,inner sep=1.5pt]{};
                    \node (t3m) at (3,1-0.5833) [circle,fill,inner sep=1.5pt]{};
                    \node (t4m) at (4,1-0.4874) [circle,fill,inner sep=1.5pt]{};
                    \node (t5m) at (5,1-0.4072) [circle,fill,inner sep=1.5pt]{};
                    \node (t1h) at (1,1-0.9176) [circle,fill,inner sep=1.5pt]{};
                    \node (t2h) at (2,1-0.7667) [circle,fill,inner sep=1.5pt]{};
                    \node (t3h) at (3,1-0.6406) [circle,fill,inner sep=1.5pt]{};
                    \node (t4h) at (4,1-0.5352) [circle,fill,inner sep=1.5pt]{};
                    \node (t5h) at (5,1-0.4472) [circle,fill,inner sep=1.5pt]{};
                    \draw[color=Vermilion, ultra thick] (start)--(t1m) node[midway, above left]{Myopic};
                    \draw[color=Vermilion, ultra thick](t1m)--(t2m)--(t3m)--(t4m)--(t5m);
                    \draw[color=Blue, ultra thick] (start)--(t1h) node[midway,below right]{Moonshot $6q$};
                    \draw[color=Blue, ultra thick] (t1h) -- (t2h)--(t3h)--(t4h)--(t5h);
                \end{axis}
                \end{tikzpicture}
        \caption{\emph{Hazard rate of science.} \footnotesize{Cumulative probability that science stops by time $t$ for different initial disclosures given $\eta=1$.}}\label{fig:hazard_rate}
    \end{figure}

While moonshots always reduce the hazard rate of science, sometimes they do not lead to a better landscape conditional on success. Suppose, for example, that $\eta$ is close to but below $\overline{\eta}$ from \cref{prop:moonshots} (and $\delta$ not too small). Then, the risk of not producing anything beyond the first discovery is large. However, a moonshot significantly reduces the risk of failure in $t=2$, in particular, if it is only marginally above the cutoff for deepening research, $\widehat{X}^\eta$. Such a moonshot does not improve the landscape, but focuses on making $R_2$ succeed. If $R_2$ completes the cycle, the knowledge landscape is worse than completing the counterfactual two steps $3q$ (through the designer) and $d^\eta(\infty)$ (through $R_2$). Yet, because $R_2$ becomes more likely to succeed after a moonshot, the moonshot remains profitable. If, instead, $\eta$ is small but above $\underline{\eta}$, a moonshot will provide benefits both in output and in the value of knowledge.

\paragraph{Taking costs into account.} One, seemingly crucial, assumption of our moonshot analysis is that the designer does not bear the costs of the initial moonshot or does not directly care about the costs exerted by future researchers. Taking these costs into account complicates the model significantly, as the designer needs to trade off several paths with potentially different costs. Analytically, we consider it beyond scope for this paper. However, numerically, it turns out that incorporating these costs and/or the costs of future researcher generations does not affect \cref{prop:moonshots}.\footnote{This and other numerical results we allude to here can be obtained using the code we provide as supplementary material to this paper. Since our static model has effectively only one (relevant) parameter, $\eta$, we can provide comprehensive numerical results.} 

\paragraph{Research cycles.} As we have seen above, with infrequent interventions, a designer may improve the evolution of knowledge by inducing a research cycle through an initial moonshot, which provides a higher payoff than the same number of steps when expanding knowledge stepwise. A natural question to ask is what happens if the designer has multiple but infrequent opportunities to intervene. There are several ways to model infrequent interventions. We opt for the simplest, which is that, in each period, the designer gets to control the researcher with probability $\lambda>0$. For the sake of clarity, we are particularly interested in the case in which such opportunities are rare. This assumption allows us to circumvent that future opportunities play a large role in the designer's decision today, which would obfuscate our trade-off of interest.

For our analytical result, we consider environments that (i) are \emph{promising} in the sense that the discount factor $\delta$ and the cost parameter $\eta$ are such that $\delta \rho^\eta(\infty)>1/2$, and (ii) feature optimal one-time moonshots of intermediate length, that is, the optimal moonshot has length $X \in [4q, \min\{2 \widehat{X}^\eta, \widetilde{X}^\eta\}]$.\footnote{Numerically, there is a wide range of parameters $(\delta,\eta)$ satisfying both conditions. \added{For example, when $\delta=0.95$, $\eta\in[0.01,0.5]$ suffices.} The lower bound, $4q$, is directly implied by the environment being promising. } Then, after a moonshot in $t$, the next researcher, $R_{t+1}$, aims to complete the cycle and chooses the midpoint of the area created by that moonshot. 

Having fixed the environment for the single moonshot, continuity implies that, for rare moonshot opportunities, the designer starts a moonshot whenever given the opportunity as long as science did not get stuck before due to failed research. However, the question is less clear when science is stuck. Specifically, we may consider the following scenario. The designer had created a moonshot in the past, but due to an unexpected turn of the truth, the researchers have failed to complete the research cycle. In that case, should the designer abandon the incomplete cycle and start a new one, or, should he use his opportunity to complete the incomplete cycle?

    \begin{proposition}\label{prop:deepen_knowledge}
        Assume \replaced{$\alpha=0$ and a promising environment with moonshots of intermediate length}{the environment is promising and features moonshots of intermediate length}. The designer strictly prefers completing an open cycle over initiating a new research cycle when interventions are sufficiently rare.
    \end{proposition}

The intuition behind \cref{prop:deepen_knowledge} is straightforward. The environment is promising and researchers are unlikely to fail while expanding research. If, in addition, the designer is patient, his effective discount factor is high. A high effective discount factor implies that leaving a cycle incomplete leads to a severe reduction in long-run research output. Moreover, the immediate payoff from completing a cycle is large. 

Providing analytical solutions for the general case proves difficult, because the researcher's problem is two-dimensional and the disclosure probabilities $\rho^\eta(\infty)$ can only be implicitly defined. Numerically, however, we can show that the logic of \cref{prop:deepen_knowledge} holds if the initial optimal moonshot is sufficiently large. The reason is intuitive. 
If the moonshot is large, there are two benefits. First, once filled up, the knowledge landscape is better than if it evolved through knowledge-expanding steps. Second, a moonshot lowers the probability that researchers get stuck \added{by increasing the success probability of the next researchers}. If, instead, the moonshot is short, the sole reason behind the moonshot is to lower the risk of researchers getting stuck. But then, if the optimal moonshot is short, researchers are likely to fail outside of moonshots. Thus, completing the cycle provides little benefits and likely leads to immediate failure thereafter. A new moonshot, on the other hand, improves the chances of a discovery in the next period. Thus, the designer opts for the latter option.

    \section[Final Remarks]{Final Remarks}\label{sec:final_remarks}

    We propose a tractable and flexible model based on three simple premises: (i) the pool of available research questions is large, (ii) questions close to existing knowledge are easier to answer than questions far from existing knowledge, and (iii) society applies knowledge when selecting policies. Our model endogenously links novelty and research output and highlights the importance of existing knowledge for research and knowledge accumulation. A dynamic extension delivers rich insights into how interventions can improve the accumulation of knowledge over time. 

    We close with a discussion along two dimensions: First, we revisit our modeling choices, and discuss extensions and alternative assumptions. Then, we discuss implications of our findings for designing institutions and research environments.

\subsection[Modeling Choices]{Modeling Choices} \label{sub:modeling_choices}

    To model the evolution of knowledge as close as possible to the static setting, we made a set of modeling assumptions that may appear strong. We now revisit these assumptions and briefly discuss alternatives.

\paragraph{Beyond the real line.} For illustrative purposes, we assume that the set of questions is the real line which implies that there are exactly two directions in which the research frontier may be expanded. In reality, we would expect that, at least in some fields, there is a plethora of directions in which a given researcher could expand the frontier. It turns out, however, that as long as the partitioning of the question space into areas is attainable, our analysis and findings remain unchanged. We provide a formal extension in Online \Cref{sub:different_universe_of_questions}. There, we also discuss the case of ``seminal discoveries'' that open up new fields of questions.
    
\paragraph{Observing failure.} In the limit, when researchers can select from a continuum of directions, it becomes hard for others to find previous failed attempts. For example, suppose that there is some cost of finding out whether a past researcher has done expanding research in the direction the current researcher contemplates. Because there is an entire ocean of potential directions, the probability that someone has worked in the contemplated direction is zero. Hence, it is not worth paying the cost and past failures remain unobserved. 
    
While our baseline model is perhaps too restrictive on the set of directions, an ocean may be too open. Thus, indeed, there are settings in which observing past failure becomes relevant. Generally, observing a failed attempt provides information about an unexpected turn of the Brownian motion and, therefore, increases the variance at that point. Assuming some coherence in the search for an answer (for example, by restricting search to a connected interval of answers), digging deeper on previous failed attempts proves unattractive: The failed attempt has revealed that the answer is complicated. In response, researchers would aim to answer a different question rather than resolving that failed attempt.

\paragraph{Getting knocked off the ladder.} \cref{prop:laissez-fair} shows that knowledge expands in steps without interventions. Because $\widehat{X}^\eta>3q$ for all $\eta$, researchers, even if heterogeneous, will not deviate from the ladder structure. Instead, with heterogeneous researchers, we would only get ladders with unequal step sizes.

Serendipity, on the other hand, may lead to researchers leaving the ladder for some time. Imagine some random discovery far from existing knowledge. Then, the next generations will work to connect it to the existing bulk of knowledge by \cref{prop:researchX} until knowledge is dense again. Only then do researchers return to the knowledge-expanding steps of \cref{prop:laissez-fair}.

Other forces with the potential to break the ladder structure are long-lived researchers who would act similar to a cost-internalizing designer and have an incentive to set moonshots, or exogenous shocks to the importance of questions. In our baseline, all questions have the same relevance. However, real-world shocks may cause an elevated interest in particular questions. Our model predicts that such temporary interest can have long-lasting effects. Once the interest fades, researchers use the improved conjectures to bridge old knowledge and new findings.

\paragraph{Other frictions.} We assume that, besides the cost friction, the market for ideas works well and researchers are paid their marginal contribution to the value of knowledge. In reality, several known frictions hinder that process. Well-known examples include publication bias \citep{Kasy2019Bias}, the emphasis on priority \citep{hill2019scooped,10.1093/restud/rdw038,hill2020race}, or researchers' career concerns \citep{Akerlof2018,10.1257/jel.20191574}. While the question of optimal market design is beyond our scope, our framework is flexible enough to incorporate these frictions. It may thus be a stepping stone toward structural models of science funding that include such frictions.

\subsection[Implications]{Implications} \label{sub:implications}

Although our dynamic model is stylized, we reconcile empirical findings in the economics of science \citep[see, for example,][documenting a lack of novelty in research]{Rzhetsky2015,Fortunato2018} and the economics of innovation \citep[e.g.,][documenting the inspiring nature of the original moonshot]{jaffe_evidence_2003}. In reality, a researcher's value of any given discovery, especially in the basic sciences, depends on the institutional framework a researcher operates in. Our findings in \cref{sub:research_dynamic} suggest that focusing only on immediate policy relevance when designing researchers' incentives is suboptimal for patient societies. In the following, we discuss some alternative incentive structures. 

\paragraph{Future-oriented rewards.} No short-lived researcher selects a question novel enough to initiate a research cycle, even when research is costless to her. The reason is that researchers exclusively care about the value of their research today but not about its indirect value in guiding future researchers. 

To incentivize moonshot discoveries, researchers need future-oriented rewards that go beyond the instantaneous value of their findings. Apart from prizes for novel findings, the value of citations for promotion decisions or scientific reputation may serve such purpose. That insight is reminiscent of recent empirical work on firm-level R\&D. \citet{frankel2023evaluation} estimate the value of dynamic spillovers from discoveries of drugs. In line with our model, they provide suggestive evidence that the lack of appropriability of these spillovers harms novelty in pharmaceutical innovation.

\paragraph{Research consortia.} The idea of research consortia has been put forward in some fields of basic science to improve the evolution of knowledge. Research consortia formed by scientists of different backgrounds operate on missions different from the ``publish or perish'' or ``marketability'' paradigms. \citet{hill2020race} document that the incentives and choices of consortia in structural biology differ from those of university researchers. In line with our model, they find that consortia provide more novel discoveries but also discoveries of lower immediate value.

Our findings suggest that the underlying reason is not particular to structural biology. Establishing institutions that alleviate some researchers from the need to provide immediate benefits can guide those in traditional incentive schemes. In fact, our findings in \cref{sub:research_dynamic} suggest that a mix in incentives may be key to improving the evolution of knowledge.

\paragraph{Direction of science.} As we noted above, funding measures that only target the researcher's cost do not alter the way knowledge progresses qualitatively. However, it may alter the direction of science. Recall that a researcher's question choice in our model can be interpreted as choosing from a set of directions and then picking novelty along that direction. Reducing the cost of research in one of the available directions distorts the marginal researcher in favor of the subsidized direction. 

Recent empirical work is consistent with this observation. \citet{nagaraj2023does} suggest that data availability affects the direction of sciences. \citet{kim2023shortcuts} finds that novel technologies can incentivize researchers to focus on more explored rather than unexplored areas. \citet{myers2020elasticity} shows that topic-specific grants can lead to a change in the direction of science. These findings can be interpreted in our model as a change in the cost of research: some directions become less cost intensive than others. 

When a researcher chooses a moonshot in some direction, she affects future generations through the implied research cycle, suggesting that the effect is persistent. Moonshots determine the direction science takes in the medium run. That coordination may provide additional network benefits, which are currently outside our model. A potential counterforce has been identified in the literature in corporate R\&D which emphasizes the role of competition.\footnote{See, e.g., \citet{bryan2017direction,hill2020race,hill2019scooped,hopenhayn2021direction}.} While beyond our scope, combining competition and dynamic spillovers within our framework is an exciting avenue for future research.

\allowdisplaybreaks

\appendix 

\section{Proofs} 
\label{sec:proofs}

Pure algebraic reformulations are relegated to \cref{sub:omitted_proofs}. We review the properties of $\tilde{c}(\rho)=(\operatorname{erf}^{-1}(\rho))^2$ in \cref{sec:Notation}. We use subscripts to denote partial derivatives; $\frac{\mathrm{d} f(x,y)}{\mathrm{d x}}$ for the total derivative; and omit function arguments when clarity is preserved.


\subsection{Proof of Proposition \ref{prop:value_knowledge}} 
\label{sub:proof_of_prop_value_knowledge}
The value of knowing $\mathcal{F}_k$ is \(
	\int_{-\infty}^{\infty} \max\left\{1- \sigma^2_x(y|\mathcal F_k)/q,\:0\right\}\mathrm{d}x.
\)
No matter which point of knowledge $(x,y(x))$ is added to $\mathcal{F}_k$, the value of knowledge outside the frontiers is identical for both $\mathcal{F}_k$ and $\mathcal{F}_k\cup\{(x,y(x))\}$. Area lengths $X_1=X_k=\infty$ do not depend on $\mathcal F_k$ and neither does the variance for a question $x<x_1$ or $x>x_k$ with a given distance $d$ to $\mathcal F_k$. The conjectures about all questions outside $[x_1,x_k]$ deliver a value of 2 $\int_0^q (q-x)/q \mathrm{d}x = q$,
which is independent of $\mathcal F_k$. 

Moreover, if a question $\hat{x}\in [x_i,x_{i+1}]$ is answered, it deepens knowledge 
and only affects questions in that area, i.e., $G(x|\mathcal F_k)=G\left(x|\mathcal F_k \cup \{(\hat{x},y(\hat{x}))\}\right)$ $\forall$ $x\notin (x_{i},x_{i+1})$.

The value of an area $[x_{i},x_{i+1}]$ is (with abuse of notation)
\begin{align*}
	v(X)=\int_0^X \max\left\{\frac{q- \frac{d(X-d)}{X}}{q},0\right\}\mathrm{d}d.
\end{align*}
Note that whenever $X\leq4q$, $d(X-d)/X\leq q$. Hence, we can directly compute the value of any area with length $X\leq 4q$ as $v(X)=X-X^2/(6q)$.
Whenever $X>4q$, a positive value is generated only on a subset of points in the area. As the variance is a symmetric quadratic function with midpoint $X/2$, there is a symmetric area around $X/2$ which has a variance exceeding $q$.  On those points, the decision-maker's losses are limited to zero. The points with variance equal to $q$ are $\overline{d}_{1,2}= X/2 \pm 1/2\sqrt{X}\sqrt{X-4q}$. Hence, the value of an area with $X>4q$ is (due to symmetry)
\begin{align*}
	v(X) &= 2\int_0^{\overline{d}_1} \frac{q-\frac{d(X-d)}{X}}{q} \mathrm{d}d = X- \frac{X^2}{6q} + \frac{X-4q}{6q}\sqrt{X}\sqrt{X-4q}
\end{align*}
If knowledge expands, a new area is created and no area is replaced. The value created is 
\begin{align*}
	V(d;\infty)=v(d)= d-\frac{d^2}{6q}  + \begin{cases}0, &\text{ if } d\leq 4q\\ \frac{d-4q}{6q}\sqrt{d}\sqrt{d-4q}, &\text{ if }d>4q. \end{cases}
\end{align*}

If a knowledge point is added inside an area with length $X$ with distance $d$ to the closest existing knowledge, it generates two new areas with length $d$ and $X-d$ that replace the old area with length $X$. The total value of the two new intervals is
\begin{align*}
	v(d)+v(X-d) = &d-\frac{d^2}{6q} &&+ \begin{cases}0, &\text{ if } d\leq 4q\\ \frac{d-4q\sqrt{d}\sqrt{d-4q}}{6q},\phantom{ABCDe}  &\text{ if }d>4q\end{cases} \\
	&+X-d -\frac{(X-d)^2}{6q} &&+ \begin{cases}0, &\text{ if } X-d\leq 4q\\ \frac{X-d-4q\sqrt{X-d}\sqrt{X-d-4q}}{6q}, &\text{ if }X-d>4q\end{cases}. 
\end{align*}

The benefits of discovery are then $V(d;X)=v(d)+v(X{-}d)-v(X)$. Noticing that $\sigma^2(d;X)=d(X{-}d)/X$ and replacing results in the expression in the proposition. Taking the limit $X \rightarrow \infty$ implies the expanding value.

\subsection{Proof of Proposition \ref{cor:opt}}
\label{sub:proof_of_cor_opt_expand}

\paragraph{Expanding Knowledge.} The first-order condition for $d\leq4q$ immediately delivers $d=3q$.
Moreover, the benefits decrease in $d$ for $d>4q$ which can be seen from the negative $d$-derivative (invoking \cref{lem:inequalityexpandinglarger4} in Online \cref{sub:omitted_proofs})
\begin{align*}
	&V_d(d;\infty|d> 4q)=-\frac{d}{3q}+1 + \sqrt{\frac{d-4q}{d}} \frac{d-q}{3q}<0,
\end{align*}


\label{sub:proof_of_cor_deep_optimal}
\paragraph{Deepening Knowledge.}~
	\begin{lemma}\label{lem:boundaryXsmaller6}
		$d^0(X)=X/2$ if $X\leq 6q$. 
	\end{lemma}
	\begin{proof}~
		\noindent\textbf{1. Assume $X \leq 4q.$} 
		The benefits of discovery are $V(d;X| X\leq 4q)=(X d - d^2)/(3q)$
		which increase in $d$ for $d\in[0,X/2]$ and are maximized at $d=X/2$. Moreover, $V(X/2;X|X\leq 4q)=X^2/(12q)$ which increases in $X$.

		\textbf{2. Assume $X \in(4q,6q]$.} Case	(i). $d \geq X-4q$ implies (since $d\leq 3q$)
		\[V(d;X)=\frac{1}{6q}\left(2dX - 2d^2-\sqrt{X}(X-4q)^{3/2}\right)\] 
		which are the same as in the first case up to the constant $-\sqrt{X}(X-4q)^{3/2}$. Thus, the optimal $d$ conditional on $d \geq X-4q$ is $d=X/2$.

		Case (ii). For $d\leq X-4q$ the benefits and their derivative are
		\begin{align*}
		V(d;X)&=\frac{2dX - 2d^2+\sqrt{X-d}(X-d-4q)^{3/2}-\sqrt{X}(X-4q)^{3/2}}{6q} \\
		V_d(d;X)&=\frac{1}{3q}\left( X-2d - (X-d-q)\sqrt{\frac{X-d-4q}{X-d}}\right)>0. 
		\end{align*}
		\cref{sub:addendum_to_sub:proof_of_cor_deep_optimal} in \cref{sub:omitted_proofs} implies the last claim. Then, $V_d(d{;}X|d{\leq} X-4q, X{\in}[4q,6q]){>}0$ $\forall$ $d$, $X$ in the considered domain; $d{=}X{-}4q$ maximizes $V(d;X|d {\leq} X-4q,X {\in}(4q,6q]))$ and by (i) $d{=}X/2$ maximizes $V(d{;}X|X {\in} (4q,6q])$.
	\end{proof}

	\begin{lemma}\label{lem:VdatXhalf}
		For any $X<\infty$, $V_d(X/2;X)=0$.
	\end{lemma}
	\begin{proof}
	\[V(d;X)=  \scriptstyle {\frac{1}{6q}} \Big(\underbrace{\scriptstyle {2X \sigma^2(d;X)}}_{(I)} + \scriptstyle {\boldsymbol{1}_{d>4q}} \underbrace{\scriptstyle {\sqrt{d}(d-4q)^{3/2}}}_{(II)}    +\scriptstyle {\boldsymbol{1}_{X-d>4q}} \underbrace{\scriptstyle {\sqrt{X-d} ~(X-d-4q)^{3/2}}}_{(III)} - \scriptstyle {\boldsymbol{1}_{X>4q}} \underbrace{\scriptstyle {\sqrt{X}(X-4q)^{3/2}}}_{(IV)}  \Big).\]
	At $d=X/2$, either both (II) and (III) are active or neither is. Moreover, (IV) is independent of $d$, and we have $\partial (II)/\partial d =-\partial (III)/\partial d$, and $\partial (I)/\partial d= 0$.
	\end{proof}

	\begin{lemma}\label{lem:dinteriorifXlarger8}
		If $X>8q$ then $d^0(X)\neq X/2$. If $d^0(X)\neq X/2$, then $d^0(X)\leq4q$. 
	\end{lemma}

	\begin{proof}
		Consider $d=\overline{d}=4q<X/2$. We obtain
		\[V(\overline{d};X|\cdot) - V(X/2;X|\cdot) =   \frac{1}{6q} \frac{(X-8q)^{3/2}}{2} \Big( 2\sqrt{X-4q} - \sqrt{X}- \sqrt{(X-8q)}\Big),\]
		which is positive if $4(X-4q)>2X-8q \Leftrightarrow X>4q$, which holds by assumption.

		To establish the second part of the lemma, note that $d>4q$ only occurs for $X>8q$. We will show that $V_d(d;X)<0$ for all $d>4q$ when $X>8q$. Towards this, observe 
		\[V_{ddX}(d;X>8q)=-\frac{4 q^2}{(X-d)^{\frac{5}{2}}(X-d-4q)^{\frac{3}{2}}}<0,\]
		because $X{>}8q$ and $X-d{\geq} X/2{>}4q$. Thus, a lower bound for $V_{dd}(d;X>8q)$ is 
		\[V_{dd}(d;X>8q)|_{\lim_{X\rightarrow \infty}} = \frac{1}{3q} \frac{d^2-d^{\frac{3}{2}}\sqrt{d-4q}-2q(d+q)} {d^{\frac{3}{2}}\sqrt{d-4q}}>0.\]
		Thus, $V_d(d;X)$ is highest for $d=X/2$ which, by \cref{lem:VdatXhalf}, is $0$. Hence, $V(d;X)$ decreases in $d$ for $X>8q$ and $d>4q$. 
	\end{proof}
	\begin{lemma}\label{lem:decreasingbenefitifXincreasing}
		$d^0(X)< X/2$ $\Rightarrow$ $\frac{d V(d^0(X);X)}{dX}<0$.
	\end{lemma}
	\begin{proof}
		By the envelope theorem, $\frac{ d V(d^0(X);X)}{d X} = V_X(d^0(X);X)$
		which is negative for $X\geq4q$ and $d\in[0,X-4q]$ by \cref{lem:partialValuetoareanegative}  in Online \cref{sub:omitted_proofs}. If $X\geq8q$, then $d\leq X-4q$ by definition and \cref{lem:partialValuetoareanegative} applies. If $X< 8q$, by \cref{lem:boundaryXsmaller6}, we know that $d^0(X)\neq X/2$ only if $X\geq 6q$. Moreover,
		\[V_d(d;X|X/2>d>X-4q, X<8q)= \frac{X-2d}{3q}>0 \]
		Hence, $d^0(X) \neq X/2 \Rightarrow d^0(X)\leq X-4q$. \cref{lem:partialValuetoareanegative} applies proving \cref{lem:decreasingbenefitifXincreasing}.
	\end{proof} 

	\begin{lemma}\label{interioronceimpliesalwaysinterior}
		$d^0(X)<X/2$ for some $X \in [6q,8q)$ $\Rightarrow$ $d^0(X)<X/2$ for all $X'>X$.
	\end{lemma}

	\begin{proof}
	It suffices to consider $X'<8q$ by \cref{lem:dinteriorifXlarger8}. We prove the claim by showing that $V(d^0_c(X);X)$ for any interior critical point $d^0_c(X)<X/2$ cuts $V(X/2;X)$ from below at \emph{any} potential intersection. Thus, there is at most one switch from $d^0(X)=X/2$ to $d^0(X)<X/2$ and no switch back. Continuity then implies the statement.

		$V(d;X)$ is a continuously differentiable function in $X$ and $d$. Thus, any interior (local) optimum $d^0_c(X)$ is continuous and so are $V(d^0_c(X);X)$ and $V(X/2;X)$. 
		We now show that if $V(d^0_c(X);X)=V(X/2;X)$ for some local optimum $d^0_c(X)<X/2$ and $X \in [6q,8q]$, then $\mathrm{d} V(d^0_c(X);X)/\mathrm{d} X > \mathrm{d} V(X/2;X)/\mathrm{d} X$. Because the first-order condition holds for an interior critical point, the envelope theorem applies and $\mathrm{d}V(d^0_c(X),X)/\mathrm{d} X <0$. By \cref{lem:boundaryXsmaller6}, we know that for $X\leq 6q$ $V(X/2,X)\geq V(d_c^0(X),X)$ if $d_c^0(X)$ exists. Hence, at the first intersection $V(d_c^0(X),X)$ must cross from above. Now, because $V(d_c^0(X),X)$ decreases, the first intersection can only occur in a region where $V(X/2,X)$ decreases and must be such that $\mathrm{d}V(X/2,X)/\mathrm{d}X<\mathrm{d}V(d^0_c(X),X)/\mathrm{d}X$. We prove that this is the only potential intersection in \cref{lem:secondtotalderivative} in Online \cref{sub:omitted_proofs}. There we show that $\mathrm{d}^2V(X/2,X)/(\mathrm{d}X)^2{<}0$ and $\mathrm{d}^2 V(d^0_c(X),X)/(\mathrm{d}X)^2{>}0$. 
	\end{proof}

	\begin{lemma}\label{lem:continuityattheiotunyn}
		$V(d^0(X);X)$ is continuous in $X$. As $X \rightarrow \infty$, it converges uniformly to $V(d;X)$ and $d^0(X) \rightarrow d^0(\infty)$. For any $X>6q$, we have $d^0(X)>3q$ and $V(d^0(X),X) > V(3q,\infty)$. 
	\end{lemma}

	\begin{proof}
		Continuity follows because $V(d^0(X);X)=\max_{d} V(d;X)$ with $V(d;X)$ continuous in both $d \in [0,X/2]$ and $X$. Now take any sequence of increasing $X_n$ with $\lim_{n \rightarrow \infty} X_n = \infty$. For any $\delta(d), \exists n$ such that $V_n(d;X_n)-V(d;\infty)<\delta(d)$ as can be seen from the formulation in the proof of Proposition \ref{prop:value_knowledge}. Hence, $V(d;X_n)$ converges uniformly to $V(d;\infty)$. By uniform convergence the maximizer $d^0(X_n)$ of $V(d;X_n)$ converges too. Convergence from above follows from $V(3q;X){>}V(3q;\infty)$ for $X{\in}(6q,\infty)$. 

		Finally, recall that $V(d;\infty)$ describes the value of an area of length $d$. That value increases for $d<3q$ and decreases for $d>3q$. Now suppose $X>6q$ and $d^0(X)<3q$. Then by increasing $d$ both areas created get closer to $3q$ and thus increase in value. A contradiction to $d^0(X)$ being the maximizer.
	\end{proof}

	\begin{lemma}\label{lem:checkX}
		$V(d^0(X);X)$ is single peaked with peak at $\widecheck{X}^0\approx 6.2q$ and $d^0(\widecheck{X}^0)\approx 3.1q$.
	\end{lemma}
	\begin{proof}
		By \cref{lem:continuityattheiotunyn}, $V(d^0(X);X)$ is continuous. By \cref{lem:decreasingbenefitifXincreasing} it is decreasing if $d^0(X)<X/2$. By \cref{lem:dinteriorifXlarger8,interioronceimpliesalwaysinterior,lem:boundaryXsmaller6} and especially the arguments used in the proof of \cref{interioronceimpliesalwaysinterior}, the (single) switch from $d^0(X)=X/2$ to $d^0(X)<X/2$ happens for some $X\in (6q;8q]$ and at a point where $V(X/2;X)$ is already decreasing. Thus, we can compute the peak by considering the first-order conditions of $V(X/2;X)$ with respect to $X$. The peak is the (real) solution to
		\begin{equation}\label{Xhalf_increasing_in_X}
			\frac{X}{X-q}=2  \frac{\sqrt{X-4q}}{\sqrt{X}}.
		\end{equation}
		Defining $m:=X/q$, the above reduces to $m/(m-1) = 2 \sqrt{(m-4)/m}$.
		For $m>4$, the LHS decreases and the RHS increases in $m$. The solution is \(m\approx 6.204.\)
	\end{proof}
	\begin{lemma}\label{lem:hatX}
		Expanding trumps deepening knowledge if and only if $X<\widehat{X}^0 \approx 4.338q$.
	\end{lemma}
	\begin{proof}
		$V(3q;X)>V(3q;\infty)$ for $X\geq 6q$ by direct comparison. For $X \in [0,6q]$ we need to consider only $d^0(X)=X/2$ by \cref{lem:boundaryXsmaller6}. We compare $V(X/2;X)$ with $V(3q;\infty)$.
		Using $m=X/q$ from the previous proof, the two functions intersect if
		\(m^2/12 - \sqrt{m}/6 (m - 4)^{(3/2)} - 3/2=0\)
		which has a unique solution such that $m\leq 6$ at $m\approx 4.338$.
	\end{proof}
	\begin{lemma}\label{lem:orderingV}
		$4q<\widehat{X}^0<6q<\widecheck{X}^0<\widetilde{X}^0<8q$.
	\end{lemma}
	\begin{proof}
		The first two inequalities follow from \cref{lem:hatX}, the third from \cref{lem:checkX}. The fourth inequality follows from \cref{interioronceimpliesalwaysinterior} and \cref{lem:dinteriorifXlarger8} implies the last.
	\end{proof}

\subsection{Proof of Proposition \ref{prop:substitutes}} 
\label{sub:proof_of_proposition_prop_subcomp}
	Throughout, we make use of the first-order necessary conditions for interior solutions.
	\noindent\begin{minipage}{.45\linewidth}
  \begin{equation}
    \eta \tilde{c}_\rho(\rho)\sigma^2(d;X) = V(d;X) \label{eq:FOCrho}\tag{FOC$^\rho$}
  \end{equation}
\end{minipage}%
\hfill
\begin{minipage}{.45\linewidth}
  \begin{equation}
    \rho V_d(d;X) =  \eta \tilde{c}(\rho) \sigma^2_d(d;X) \label{eq:FOCd}\tag{FOC$^d$}
  \end{equation}
\end{minipage}
	\paragraph{Part 1: Expanding Knowledge} 
		\begin{lemma}\label{lem:interior_existence}
			There is a non-trivial optimal choice with $\infty>d>0, 1>\rho>0$ on any interval with positive length, $X\in \mathbb{R}^{+} \cup \{\infty\}$. The first-order condition, \eqref{eq:FOCrho}, is necessary for optimality of $\rho^\eta(X)$.
		\end{lemma}
		\begin{proof}
			The researcher can guarantee a non-negative payoff by choosing $d=0$ or $\rho=0$. Hence, her value is bounded from below by zero
			. Next, note that $u_R(\rho=0,d>\varepsilon;X)=0$ for any $\varepsilon>0$ and that $(\partial u_R(\rho,d;X))/(\partial \rho)|_{\rho=0,d=\varepsilon}=V(\varepsilon,X)>0$ by \Cref{prop:value_knowledge}. Therefore, on any interval $X$ a maximum with $d>0,\rho>0$ exists. 

			By \cref{lem:checkX}, the benefits of discovery are bounded $V(d,X)\leq M<\infty$ and $\lim\limits_{\rho \rightarrow 1} \tilde{c}(\rho)=\infty$. Therefore, the optimal $\rho<1$. Finally, $V(d,\infty)$ decreases in $d$ for $d$ large enough while the cost $\eta \tilde{c}(\rho) \sigma^2(d,\infty)$ increases in $d$. Hence, the optimal distance is bounded $d^\eta(\cdot)\leq D <\infty$.

			Because the optimal choice is interior and the objective is continuously differentiable, a necessary condition for $\rho^\eta(X)$ is that it solves \eqref{eq:FOCrho}. 
		\end{proof}


		\begin{lemma}\label{lem:expanding_focs}
				When expanding knowledge, the optimal choice is characterized by the first-order conditions \eqref{eq:FOCd} and \eqref{eq:FOCrho}. These FOCs suffice and $d^\eta(\infty)\in(2q,3q)$. The researcher's value is strictly positive $U_R(\infty)>0$.
			\end{lemma}
			\begin{proof} 
				We proceed in three steps. First, we show that the distance is at most $3q$. Second, we show that the first-order conditions are sufficient when expanding knowledge. Third, we characterize the optimal choice of the researcher.

				\emph{Step 1. $d\leq 3q$.} Fix any $\rho\geq 0$. Since $\sigma^2(d;\infty)$ increases in $d$, it is immediate that the researcher's utility is non-increasing in $d$ if $V(d;\infty)$ decreases in $d$. Combining this observation with \cref{cor:opt}, it is sufficient to restrict attention to $d\leq 3q$.

				\emph{Step 2. FOCs sufficient.} By \Cref{lem:interior_existence}, the researcher's optimal choice is interior and, hence, characterized by the first-order conditions. To see the sufficiency of the first-order conditions, note that the first principal minor of Hessian is $\rho V_{dd} - \eta c \sigma^2_{dd}=-\rho \frac{1}{3q}<0$ as $\sigma^2_{dd}=0$ and that the second principal minor is given by the determinant of the Hessian at the critical point:
				\begin{align}\label{hesse1}
					&-\rho V_{dd}(d;\infty)\eta  \tilde{c}_{\rho \rho}(\rho) \sigma^2(d;\infty)-(V_d-\eta \tilde{c}_{\rho}(\rho)\sigma^2_d(d;\infty))^2\notag\\
					=&\rho \frac{\tilde{c}_{\rho \rho}(\rho)}{\tilde{c}_{\rho}(\rho)} \frac{V(d;\infty)}{3q} - \left(-\frac{d}{3q}+1-\frac{V(d;\infty)}{\sigma^2(d;\infty)}\right)^2.
				\end{align}
				The equality follows from $V_{dd}=-\frac{1}{3q}$, $\sigma^2(d;\infty)=d$, and using the necessary condition \eqref{eq:FOCrho} via $\eta \sigma^2(d;\infty)= V(d;\infty)/\tilde{c}_\rho(\rho)$.
				Substituting for $V(d;\infty)=d-d^2/(6q)$ (as $d\leq 3q$ by Step 1) yields as condition for a positive second principal minor:
				\begin{align*}
					\rho \frac{\tilde{c}_{\rho \rho}(\rho)}{\tilde{c}_{\rho}(\rho)}&>\frac{d}{2(6q-d)}.
				\end{align*}
				The inequality holds as the properties of $\tilde{c}$ imply $LHS\geq 1$ while $RHS \leq \frac{1}{2}$ for $d\leq 3q$.

				\emph{Step 3. Characterization.} Substituting the expressions for $V(d;\infty)$ and $\sigma^2(d;\infty)$ for expanding knowledge into the first-order condition \eqref{eq:FOCd} yields $\rho(1- d/(3q)) = \eta \tilde{c}(\rho).$
				Replacing $\eta$ via \cref{eq:FOCrho} and solving for $d$ we obtain
				\[d^\eta(\infty)=3q \left(1- \frac{\tilde{c}(\rho)}{2 \tilde{c}_\rho(\rho) \rho- \tilde{c}(\rho)}\right) \in (2q,3q)\]
				where the bounds follow from the properties of $\tilde{c}$.
			\end{proof}

\paragraph{Part 2. Deepening knowledge.}~\medskip

	\begin{lemma}\label{lem:single_crossing}
		The researcher's optimal choice of distance is 
		$d^\eta(X)=\frac{X}{2}$ for $X\leq \widetilde{X}^\eta$ and 
		$d^\eta(X)<X/2$ otherwise. At $\widetilde{X}^\eta$, the payoff $U_R(X)$ decreases. Further, $\lim\limits_{X\rightarrow \infty}d^\eta(X)=d^\eta(\infty)$ and the convergence is from above. Any optimal distance satisfies $d^\eta(X)\leq 4q$.
	\end{lemma}
	\begin{proof}
		Define $d^b:=X/2$---the boundary solution---, and $d^i$ as the solution $d$ to \eqref{eq:FOCd} assuming $d<X/2$ (whenever it exists)---the interior solution.\medskip 

		\emph{Step 1. $d^b$ always a candidate solution.} Note first that the choice $d^b$ always constitutes a local maximum as the marginal cost of distance is zero at this point, $\partial \sigma^2 (d,X)/\partial d=1-2d/X$. Moreover, by \cref{lem:VdatXhalf} also the marginal benefit is zero at $d=X/2$. Finally, for any choice of $d$, there is a unique $\rho$ that solves \eqref{eq:FOCrho} because, given $d$, \eqref{eq:FOCrho} has a continuous, strictly increasing, left-hand side with $\tilde{c}_{\rho}(0)=0$ and $\lim_{\rho \rightarrow 1}\tilde{c}_\rho(\rho)=\infty$, and has a constant right-hand side. Hence, the boundary solution with $d^b$ is always a candidate solution. 

		\emph{Step 2. $d^\eta(X)=X/2$ if $X\leq 4q$.}	Recall $\eqref{eq:FOCrho}$ and $\eqref{eq:FOCd}$. Assuming an interior solution $d^i$, replacing $\eta$ via \eqref{eq:FOCrho} in \eqref{eq:FOCd} we obtain 
		\begin{align*}
			 	{\frac{V_d(d,X)}{\sigma^2_d(d,X)}}\Bigg/{\frac{V(d,X)}{\sigma^2(d,X)}}={\frac{\tilde{c}(\rho)}{\rho}}\Big/{\tilde{c}_{\rho}(\rho)}.
		\end{align*} 
		It follows from the properties of $\tilde{c}(\rho)$ that the $RHS\in[0,1/2]$ and decreasing. Thus, if the $LHS>1/2$, it is beneficial to increase $d$ if possible and the boundary choice $d^b$ is optimal. For short areas, $X\leq 4q$, the boundary choice is optimal as
		\begin{align*}
			{\frac{V_d}{\sigma^2_d}}\Bigg/{\frac{V}{\sigma^2}} ={\frac{2(X-2d)}{\frac{X-2d}{X}}}\Bigg/{\frac{2(dX-d^2)}{\frac{d(X-d)}{X}}}=1.
		\end{align*}

		\emph{Step 3. $d^\eta(X)<X/2$ if $X> 8q$.} Note first that the variance of the boundary question is always greater than for any interior question. 
		Hence, if the benefits of a discovery, $V$, are larger for an interior question than for the boundary question, the researcher can obtain a higher payoff by choosing an interior question with the same $\rho$ as for the boundary question. 
		The benefits of discovery on the boundary of an area with $X>8q$ are always smaller than for some interior distance by \cref{lem:dinteriorifXlarger8}. Hence, an interior choice is optimal for $X>8q$.

		\emph{Step 4. If $d^i$ is optimal it must be that $d^i<4q$ and that $X-d^i>4q$.} 
		
				For $X\in (4q,8q)$ and $X-d<4q$, 
				\begin{align*}
					{\frac{V_d(d,X)}{\sigma^2_d(d,X)}}\Bigg/{\frac{V(d,X)}{\sigma^2(d,X)}}=\frac{2d(X-d)}{-2d^2+2dX-\sqrt{X}(X-4q)^{3/2}}
				\end{align*}
				which decreases in $d$ with $\lim_{d \rightarrow X/2} ~ X^2/2 /(X^2/2-\sqrt{X}(X-4q)^{3/2}) $,
				which, in turn, increases in $X$ and is one for $X=4q$. Hence, any interior solution must be such that $X-d>4q$ for the same reasons given in Step 2 of this proof. For $X-d<4q$, $\partial u_R(\rho,d;X)/\partial $ 
				is always positive. 
				
				For $X>8q$, $X-d^i>4q$. $d^i<4q$ follows as $V(d;X)$ decreases in $d$ when $d>4q$ by \cref{lem:dinteriorifXlarger8}.

		\emph{Summary Step 1-4.} We know that (i) in areas with $X<4q$, the researcher's distance choice on the deepening area will be $d^b$, (ii) in areas with $X>8q$ the researcher's distance choice will be $d^i$, (iii) in areas with $X\in[4q,8q]$ the researcher's distance choice may $d^i$ or $d^b$, but (iv) if the solution is $d^i$, it has to satisfy $X-d^i>4q$ and $d^i<4q$. The latter two imply $d^i<X/2$ in this case.
		\medskip

		\emph{Step 5. Single crossing of the payoffs.} With three observations, we show that the payoffs, $U_R(d^b;X)$ and $U_R(d^i;X)$, cross once assuming $\rho(d,X)$ is chosen optimally. 
		\begin{enumerate}
			\item At area length $X$ for which $U_R(d^b;X)=U_R(d^i;X)$, the payoff at the boundary solution must be decreasing faster than at the interior solution. 
			\item On the interval $[4q,8q]$, the payoff of the boundary solution has a strictly lower second derivative with respect to $X$ for all $X$ than that of the interior solution. Hence, the two values can cross at most once on this interval.
			\item $U_R(d^b;X) \leq U_R(d^i;X)$ if $X\geq 8q$.
		\end{enumerate} 

		The first observation follows because the first switch is from the boundary to the interior solution, by continuous differentiability of all terms and $d^\eta(X)=X/2$ for $X<4q$. The third observation follows from Step 3 above.

		The second observation follows from totally differentiating $U_R$ for both local maxima. Using envelope conditions, we obtain that the payoff is concave in the boundary solution and convex in the interior solution implying the second observation. Define $\varphi(X):=\max_\rho u(d=X/2,\rho,X)$ for the boundary; we show in \cref{lem:URconcaveBound} that $\varphi(X)$ is concave. In \cref{URconvexIN}, we show that $U_R(X)=\max_{\rho,d} u(d,\rho,X)$ is convex in $X$ if the maximizer satisfies $d^\eta(X)<X/2$. 

		\emph{Step 6. Asymptotics.} 
		As $X\rightarrow \infty$, $V(d,X)$ converges to $V(d,\infty)$ and $\sigma^2(d,X)$ to $\sigma^2(d,\infty)$ and the researcher's optimization on the deepening interval converges to that on the expanding interval which has a unique maximum at $(d^\eta(\infty),\rho^\eta(\infty))$. In particular, if such an optimum exists, the envelope condition implies that $\mathrm{d} U_R(d^i(X);X)/\mathrm{d} X=\rho V_X(d^i,X) - \eta \tilde{c}(\rho)\sigma^2_X(d^i,X)<0$
		as $V_X(d,X)<0$ according to \cref{lem:partialValuetoareanegative} for $X>4q$ and $X-d>4q$ and $\sigma^2_X(d,X)>0$. Hence, the payoff of any optimal interior choice decreases in $X$.
	\end{proof} 
\subsection{Proof of Proposition \ref{prop:researchX}}
\label{sub:proof_of_prop_researchX}
We prove the statements in \cref{prop:researchX} in reverse order. A side-product of this proof is that we show: $4q<\widecheck{X}^\eta\leq \widecheck{X}^0$, $\dot{X} \approx 4.548q$ and $\widehat{X}^\eta< \widehat{X}^0$.
	\paragraph{Step 1: Proof of \Cref{item:check}.}
		We use a series of lemmata to show that a local maximum, $\widecheck{X}^\eta$, exists (\cref{lem:decreasing_if_V_decreasing,lem:boundary_peak}) and that it is global (\cref{lem:single_peaked}).
		\begin{lemma}\label{lem:decreasing_if_V_decreasing}
			Fix $d=X/2$ and assume that an interior optimum exists. Then $U_R(X|d=X/2)$ is maximal only if the total differential $\mathrm{d} V(d=X/2;X)/\mathrm{d} X\geq 0$.
		\end{lemma}

		\begin{proof}
			Under the assumption that $d=X/2$, $U_R(X)$ is defined and continuously differentiable for all $X \in [0,\infty)$ despite the indicator functions.\footnote{Note that the terms appearing in the indicator functions are of the form $\sqrt{a}(a-4q)^{3/2}$. Taking the limit of their derivative from above to $4q$ yields zero such that the left and right derivative coincide at the point at which the indicator functions become active.} Because $X=0$ implies $U_R(X=0)=0$, because $U_R(X)$ declines for $X$ large enough and because \cref{lem:interior_existence} holds, there is an interior $X$ at which $U_R(X)$ is maximized.

			As $U_R(X)$ is differentiable and maximized at some interior $X$ it satisfies $\partial U_R/\partial X = 0$. By assumption, $d(\widecheck{X}^\eta)=X/2$ and \eqref{eq:FOCrho} holds. Thus,
		\(\rho \frac{\mathrm{d} V(d=X/2;X)}{\mathrm{d} X}= \frac{\eta}{4} \tilde{c}(\rho).\)
		The right-hand side is non-negative, implying the desired result.\footnote{The RHS is only $0$ if $\eta=0$, $\rho^\eta(X)=1$ and $U_R(X)=V(X)$.} 
		\end{proof}

		\begin{lemma}\label{lem:boundary_peak}
			$U_R(X;d^b=X/2)$ peaks at $X=\widecheck{X}^\eta\in(4q,\widecheck{X}^0]$.
		\end{lemma}
		\begin{proof}
		Define $\widehat{U}_R(X)=U_R(X;d^b= X/2)$. Note that $\widehat{U}_R'(X)>0$ for $X\in [0,4q]$. This follows because in this case $\widehat{U}_R(X)=\rho X^2/(12q)-\eta \tilde{c}(\rho)X/4$ and, hence, $\widehat{U}_R'(X)=\rho X/(6q)-\eta \tilde{c}(\rho) /4$. Using optimality of $\rho$ via the \eqref{eq:FOCrho}, we obtain $X/(6q)=\eta \tilde{c}_{\rho}(\rho)/2$
			which yields $\widehat{U}'_R(X) =\tilde{c}_{\rho}(\rho)\rho \eta /4  \left(2-\tilde{c}(\rho)/(\rho \tilde{c}_{\rho}(\rho)) \right) >0$,
			where the inequality follows again from the properties of $\tilde{c}(\rho)$.

			Moreover, $\widehat{U}_R(X)$ is strictly concave on $[4q,8q]$ as $\widehat{V}(X):=V(d=X/2,X)$ is concave by \cref{lem:secondtotalderivative} and $\mathrm{d}\mathrm{d}\sigma^2(d=X/2,X)/(\mathrm{d}X\mathrm{d}X)=0$ implying $\widehat{U}_R''(X)=\rho \widehat{V}_{XX}<0$.\footnote{Where $\rho'(x)=0$ by optimality and the property of the first-order condition.} For $X>\widecheck{X}^0$, $\mathrm{d} V(d=X/2;X)/\mathrm{d} X<0$ by the definition of $\widecheck{X}^0$ implying that for $X>\widecheck{X}^0$ the researcher's value decreases. By \cref{lem:decreasing_if_V_decreasing}, it follows that the value-maximizing area length $\widecheck{X}^\eta\in(4q,\widecheck{X}^0]$.
		\end{proof}

		\begin{lemma}\label{lem:single_peaked}
			$U_R(X)$ is single peaked in $X$ with the maximum attained at $\widecheck{X}^\eta$.
		\end{lemma}
		\begin{proof}
			The result follows from three observations: First, $\widetilde{X}^\eta>\widecheck{X}^\eta>4q$ by \cref{lem:single_crossing,lem:boundary_peak}. Second, $V_X(d;X)<0$ if $X>4q$ and $d<X/2$ by \cref{lem:decreasingbenefitifXincreasing}. Third, by the envelope theorem, if $d^\eta(X)<X/2$, then $\partial U_R(X)/\partial X= \rho^\eta(X) V_X(d^\eta(X);X) - \eta\tilde{c}(\rho^\eta(X)) \sigma_X(d^\eta(X);X)<\rho^\eta(X) V_X(d^\eta(X);X)$.
			Thus, the interior-solution payoff intersects the boundary-solution payoff from below while both decrease. 
		\end{proof}

	\paragraph{Step 2. Proof of \Cref{item:dot}.}~

		\noindent\underline{Step 2.1 Maximum of $d^\eta(X)$ at $\widetilde{X}^\eta$.} By \cref{lem:single_crossing}, $d^\eta(X)$ increases for $X<\widetilde{X}^\eta$. By Step 4 in the proof of \cref{lem:single_crossing}, we know that any interior solution $d^i$ is such that $d^i<4q<X-d^i$ and thus strictly smaller than $X/2$. Thus, $d^\eta(X)$ decreases when it switches from the boundary to an interior solution. 

		\noindent\underline{Step 2.2 Maximum of $\rho^\eta(X)$ at $\dot{X}$.} We guess (and verify in Step 4) that a maximum of $\rho^\eta(X)$ exists in the range $[\widehat{X}^\eta,\widetilde{X}^\eta]$, that is the region in which it is optimal to deepen knowledge and to select $d=X/2$.
		\begin{lemma}\label{lem:dotX}
			Suppose $d=X/2$ is optimal for a range \([\underline{X},\overline{X}]\) such that $d^\eta(X)=X/2$. Then, the optimal $\rho^\eta(X)$ is single peaked. It is highest at $\dot{X}=8 \cos\left(\pi/18\right)/\sqrt{3}$.
		\end{lemma}
		\begin{proof}
			By \cref{lem:decreasing_if_V_decreasing}, we know that $\mathrm{d} V(d=X/2;X)/\mathrm{d} X \geq0$ and, by \cref{lem:boundary_peak}, $\overline{X}>\widehat{X}^0$. Moreover, recall $\sigma^2(d=X/2;X)=X/4$. The first-order condition with respect to $\rho$ becomes $V(X/2;X)/X= \eta\tilde{c}_\rho(\rho)/4$,
			with $V(X/2;X)/X= X/(12 q) - \textbf{1}_{X>4q} (X-4q)^{3/2}/(\sqrt{X} 6q)$.

				The latter is continuous and concave. Since $\tilde{c}(\rho)$ is an increasing, twice continuously differentiable and convex function, $\rho$ increases in $X$ if and only if $V(X/2;X)/X$ increases in $X$. By concavity of $V(X/2;X)/X$ that implies single peakedness.

			Thus, $\dot{X}$ is independent of $\eta$ and given by
			 \(\dot{X} =8 \cos\left(\pi/18\right)/\sqrt{3}\approx 4.548 q. \)
		\end{proof}

	\paragraph{Step 3. Proof of \Cref{item:hat}.} 
		\begin{lemma}\label{lem:hatXsmallerhatX0}
			$\widehat{X}^\eta$ exists, $\lim_{X \searrow \widehat{X}^\eta} \rho^\eta(X) >\rho^\eta(\infty)$, and $\widehat{X}^\eta$ decreases in $\eta$.
		\end{lemma}
		\begin{proof}
			As $X \rightarrow 0$, $d^\eta(X) \rightarrow 0$ and thus $U_R(X) \rightarrow 0$. By \cref{lem:expanding_focs}, $U_R(\infty)>0$. Thus, by continuity of $U_R(X)$, $\exists \widehat{X}^\eta>0$ such that expanding research dominates deepening research for all $X<\widehat{X}$. Cost are increasing in $X$ and by \cref{cor:maxinteriorsmaller8}, $V(d;X\in (\widehat{X}^0,\infty))>V(d;\infty)$ which implies $U_R(X\in(\widehat{X}^0,\infty))>U_R(\infty)$. By \cref{lem:single_peaked} and again continuity of $U_R(X)$, that payoff is maximal at $\widecheck{X}$. Thus, we obtain that $\widehat{X^\eta}$ exists and that $\widehat{X}^\eta<\widecheck{X}^\eta$. 

			We now show that $\lim_{X \searrow \widehat{X}^\eta} \rho^\eta(X) >\rho^\eta(\infty)$ holds if $\widehat{X}^\eta<6q$, then we show that $\widehat{X}^\eta$ decreases in $\eta$ which, together with the observation that $\widehat{X}^0<6q$, proves the lemma. 
			At $\widehat{X}^\eta$ we have $U_R(\widehat{X})=U_R(\infty)$:
			\begin{equation} \label{indifferencehatX}
			\rho(\widehat{X}^\eta) V(\widehat{X}^\eta/2;\widehat{X}^\eta) - \eta \tilde{c}(\rho(\widehat{X}^\eta)) \widehat{X}^\eta /4 =\rho^\eta(\infty) V(d^\eta(\infty);\infty) - \eta \tilde{c}(\rho^\eta(\infty)) d^\eta(\infty),
			\end{equation}
			where the fact that $d(\widehat{X}^\eta)=\widehat{X}^\eta/2$ follows from \cref{lem:single_peaked,lem:single_crossing,lem:boundary_peak}. Moreover, the following holds by optimality
			\begin{align}
				V(d^\eta(\infty);\infty) &= \eta \tilde{c}_\rho(\rho^\eta(\infty))d^\eta(\infty)\tag{FOC $\rho^\eta(\infty)$}\\
				V(\widehat{X}^\eta/2;\widehat{X}^\eta) &= \eta \tilde{c}_\rho(\rho(\widehat{X}^\eta))\frac{\widehat{X}^\eta}{4}.\tag{FOC $\rho(\widehat{X})$}
			\end{align} 

			\noindent\underline{Claim 1: $\rho^\eta(\infty)<\rho(\widehat{X}^\eta)$ if $\widehat{X}^\eta<6q$.}
				Using (FOC $\rho^\eta(\infty)$) and (FOC $\rho(\widehat{X}^\eta)$), we obtain from the properties of the error function $\rho(\widehat{X}^\eta)>\rho^\eta(\infty)$ if and only if
				\[4 \frac{V(\widehat{X}^\eta/2;\widehat{X}^\eta/2)}{\widehat{X}^\eta} > \frac{V(d^\eta(\infty);\infty)}{d^\eta(\infty)}.\]
				\emph{Case 1: $\widehat{X}^\eta>4q$.}
				Substituting for the $V(\cdot)$'s the above becomes $d^\eta(\infty)+ 2\widehat{X}^\eta - 4 (\widehat{X}^\eta-4q)^{3/2}/\sqrt{\widehat{X}^\eta}>6q$.
					A sufficient condition for this to hold is 
					\( d^\eta(\infty) -2 \widehat{X} + 10q>0\). 
					Using that $d^\eta(\infty)>2q$ by \cref{lem:expanding_focs}, we obtain that a sufficient condition for $\rho(\widehat{X}^\eta)>\rho^\eta(\infty)$ is $\widehat{X}^\eta<6q$. 

				\emph{Case 2: $\widehat{X}^\eta \in (2q,4q]$.} Performing the same steps assuming that $\widehat{X}^\eta \in [2q,4q]$, the claim holds if and only if \(\widehat{X}^\eta/(3q) > 1- d^\eta(\infty)/(6q) \:
						\Leftrightarrow \: 2\widehat{X}^\eta > 6q - d^\eta(\infty)> 4q\)
					implying the result.

				\emph{Case 3: $\widehat{X}^\eta<2q$.} We show that Case 3 never occurs, that is $\widehat{X}^\eta>2q$. To do so, we compare $U_R(d=2q;\infty)$ with $U_R(d=1q;X=2q)$ and show that the former is always larger. Hence, $X=2q<\widehat{X}^\eta$ for any $\eta$. For $X=d=2q$, we have \(X/(3q) = 1- d/(6q)\), and thus $\rho(X=2q)=\rho(d;\infty)=\rho$ (cf. Case 2). Moreover, we have
					\(V(1q;2q)=q/3\) and \(V(2q;\infty)= 4q/3\). (FOC $\rho^{X}$) implies
					\(4 V(1q;2q)/2q = 2/3 = {\eta} \tilde{c}_\rho(\rho)\). Since $\tilde{c}_\rho(\rho)>\tilde{c}(\rho)/\rho $ for any $\rho>0$, these imply $\eta \tilde{c}(\rho)/\rho<2/3$.
					Note that $U_R(d=2q;\infty) - U_R(X=2q) = q \left(\rho - 3/2\eta \tilde{c}(\rho)\right)$,
					which is positive as $\eta \tilde{c}(\rho)/\rho<2/3$. Thus, $U_R(d=2q;\infty)>U_R(X=2q)$ and therefore $\widehat{X}^\eta>2q$.\medskip

			\noindent\underline{Claim 2: If $\rho^\eta(\infty)<\rho(\widehat{X}^\eta)$ then $\widehat{X}^\eta$ decreases in $\eta$.}

				Using (FOC $\rho^\eta(\infty)$) and (FOC $\rho(\widehat{X}^\eta)$) to replace $V(\cdot)$s in \cref{indifferencehatX}, we obtain
				\(d^\eta(\infty)\left(\rho^\eta(\infty) \tilde{c}_\rho(\rho^\eta(\infty)) - \tilde{c}(\rho^\eta(\infty))\right) = \widetilde{X}^\eta/4 \left(\rho(\widehat{X}^\eta)  \tilde{c}_\rho(\rho(\widehat{X}^\eta)) -\tilde{c}(\rho(\widehat{X}))\right)   \),	yielding
				\[\widehat{X}^\eta/4 = d^\eta(\infty) \frac{\left(\rho^\eta(\infty) \tilde{c}_\rho(\rho^\eta(\infty)) - \tilde{c}(\rho^\eta(\infty))\right)}{\left(\rho(\widehat{X}^\eta)\tilde{c}_\rho(\rho(\widehat{X}^\eta)) -\tilde{c}(\rho(\widehat{X}^\eta))\right)}.\]

				The envelope theorem gives
				\[\frac{\partial }{\partial \eta}\left(U_R(\widehat{X}^\eta)- U_R(\infty)\right) = -\tilde{c}(\rho(\widehat{X}^\eta)) \frac{\widehat{X}^\eta }{4} + \tilde{c}(\rho^\eta(\infty)) d^\eta(\infty).\]
				Replacing for $\widehat{X}^\eta$ implies that the RHS is positive if and only if
				\[\left(\tilde{c}(\rho^\eta(\infty))\right)- \tilde{c}(\rho(\widehat{X}^\eta))\frac{\rho^\eta(\infty) \tilde{c}_\rho(\rho^\eta(\infty))-\tilde{c}(\rho^\eta(\infty))}{\rho(\widehat{X}^\eta) \tilde{c}_\rho(\rho(\widehat{X}^\eta))-\tilde{c}(\rho(\widehat{X}^\eta))}>0.\]
				Using that $\rho\tilde{c}_\rho(\rho)>\tilde{c}(\rho)$ by the properties of the inverse error function and factoring out the denominator of the first term, the above holds if and only if
				\[\tilde{c}(\rho^\eta(\infty)) \rho^\eta(\infty) \tilde{c}_\rho(\rho(\widehat{X}^\eta)) - \tilde{c}(\rho(\widehat{X}^\eta) \rho^\eta(\infty) \tilde{c}_\rho(\rho^\eta(\infty)) >0 \Leftrightarrow
				\frac{\rho(\widehat{X}^\eta) \tilde{c}_\rho(\rho(\widehat{X}^\eta))}{\tilde{c}(\rho(\widehat{X}^\eta))}>\frac{\rho^\eta(\infty) \tilde{c}_\rho(\rho^\eta(\infty))}{\tilde{c}(\rho^\eta(\infty))}\]
				which holds if and only if $\rho(\widehat{X}^\eta)>\rho^\eta(\infty)$ by the properties of the error function. Thus, $\widehat{X}^\eta$ decreases if $\rho(\widehat{X}^\eta)>\rho^\eta(\infty)$.

			\noindent\underline{Conclusion:} Since $\widehat{X}^0\in [2q,6q]$, $\rho^\eta(\infty)<\rho(\widehat{X}^\eta)$ implying that $\widehat{X}^\eta$ decreases in $\eta$. 
		\end{proof}

	\paragraph{Step 4.} 
	\begin{lemma}\label{lem:ordereta}
	 $\widehat{X}^\eta<\dot{X}<\widecheck{X}^\eta<\widetilde{X}^\eta$.
	\end{lemma}\begin{proof}
		We first show that $\widecheck{X}^\eta>\dot{X}$. 
		By the envelope theorem, we need at $X=\widecheck{X}^\eta$ that 
		$\rho\: \mathrm{d} V(d=\widecheck{X}^\eta/2;\widecheck{X}^\eta)/\mathrm{d} X = \eta \tilde{c}(\rho)/4$.
		The FOC for $\rho$ at $X=\widecheck{X}^\eta$ gives
		\(V/\widecheck{X}^\eta =  \eta \tilde{c}_{\rho}(\rho)/4\).
		Assume for a contradiction that $\rho(\widecheck{X}^\eta)$ increases. Then, $V/\widecheck{X}^\eta$ must be increasing which holds if and only if \(\mathrm{d} V(d=\widecheck{X}^\eta/2;\widecheck{X}^\eta)/\mathrm{d} X > V(d=\widecheck{X}^\eta/2;\widecheck{X}^\eta)/ \widecheck{X}^\eta\).
		Combining this inequality with the above properties  leads to a contradiction as $\tilde{c}_{\rho}(\rho)>\tilde{c}(\rho)/\rho$ by the properties of $\tilde{c}(\rho)$.
	
		\noindent\underline{Ordering.}
			By \cref{lem:hatXsmallerhatX0} we know that $\widehat{X}^\eta<\widehat{X}^0$. Thus, because $\widehat{X}^0<\dot{X}$ $\Rightarrow$ $\widehat{X}^\eta<\dot{X}$. Moreover, $\widetilde{X}^\eta>\widecheck{X}^\eta$ by \cref{lem:single_crossing} which concludes the proof.\end{proof}


\subsection{Proof of Proposition \ref{prop:laissez-fair}} 
\label{sub:proof_of_proposition_prop:laissez-fair}

By assumption, knowledge is dense initially and the base step is satisfied. We show the induction step assuming $R_t$ chooses $d^\eta(\infty)$. We have to show that $\widehat{X}^\eta>d^\eta(\infty)$ to show that knowledge is dense in $t+1$. Suppose the opposite holds, then from \eqref{eq:FOCrho} we know
\((6q-d^\eta(\infty))/(6q)=V(d^\eta(\infty);\infty)/\sigma^2(d^\eta(\infty);\infty)=\eta \tilde{c}_\rho(\rho^\eta(\infty))\) and \(d^\eta(\infty)/(6q)=V(d^\eta(\infty)/2;d^\eta(\infty))/\sigma^2(d^\eta(\infty)/2;d^\eta(\infty))=\eta \tilde{c}_\rho(\rho(d^\eta(\infty)))\)
implying
\[\rho^\eta(\infty) 
= \operatorname{erf}\left(\begin{gathered}\sqrt{\frac{W\left(\frac{36q^2-12qd^\eta(\infty)+(d^\eta(\infty))^2}{18q^2 \eta^2 \pi}\right)}{2}}\end{gathered}\right) \textrm{ and } \rho(d^\eta(\infty))=\operatorname{erf}\left(\sqrt{\frac{W\left(\frac{(d^\eta(\infty))^2}{18q^2 \eta^2 \pi}\right)}{2}}\right),\]
where $W(\cdot)$ is the Lambert W function. 
Because $d^\eta(\infty)$ is linear in $q$ by \cref{homoq} and $d^\eta(\infty)<3q$ by \cref{lem:expanding_focs}, it follows that $36 q^2-12qd^\eta(\infty)>0$ which implies that $\rho^\eta(\infty)>\rho(d^\eta(\infty))$ by the monotonicity of the Lambert W function.

By \cref{lem:dotX}, we know that $\rho^\eta(X)$ increases for $X<\dot{X}=8 cos(\pi/18)/\sqrt{3}$. By \cref{lem:hatXsmallerhatX0}, we know $\rho(\widehat{X}^\eta)>\rho^\eta(\infty)$. By \cref{lem:ordereta}, $\widehat{X}^\eta<\dot{X}$. Thus, $d^\eta(\infty)<\widehat{X}^\eta$ which implies that if $R_t$ expands knowledge, so does $R_{t+1}$. 

The proof of the comparative statics result follows directly from the implicit function theorem and the first-order conditions; see Online \Cref{sec:cs_expanding} for details.

\subsection{Proof of Proposition \ref{prop:moonshots}} 
\label{sec:proof_of_prop:moonshots}

We begin with the negative benchmark results for $\eta \rightarrow \infty$ and $\eta=0$. If $\eta \rightarrow \infty$, research becomes infinitely costly. Hence, $\rho \rightarrow 0$ and $u_R(d^\eta(X);\rho^\eta(X);X) \rightarrow 0$ for any $X$. Thus, absent interventions, research creates no value. Any disclosure by the designer should maximize the immediate payoff $V(d;\infty)$.

If $\eta=0$, $\rho^0(\cdot)=1$ and $u_R(d,\rho^0;X)=V(d;X)$. Thus, each researcher maximizes $V(d;X)$. By construction, maximizing the per-period $V(d;X)$ corresponds to maximizing the long-run objective of the decision-maker. 

For intermediate ranges, it suffices to show that selecting a moonshot of length $6q$ is preferred to selecting the myopically optimal $d=3q$ for some $(\underline{\eta},\overline{\eta})$ and $\delta(\eta)<1$. 

We restrict attention to $\eta$-levels such that $d^\eta(6q)=3q.$ These levels exist due to $\widetilde{X}^0>6q$ by \cref{lem:orderingV}. A moonshot is preferred if and only if
\[\frac{V(6q;\infty) +  \delta\rho^\eta(6q)\Big(V(3q;6q) +\delta  CV\Big)}{1-\delta} \geq \frac{V(3q;\infty) + \delta \rho^\eta(\infty) \Big(V(d^\eta(\infty);\infty) + \delta CV\Big)}{1-\delta},\]
where $CV$ is the (common) continuation value conditional on researchers not being stuck at $t=3$. Because $\rho^\eta(6q)>\rho^\eta(\infty)$ by \cref{prop:substitutes} a sufficient condition for the above is that
\begin{equation}\label{eq:moonshotrelevantineq}
	V(6q;\infty) +  \delta\rho^\eta(6q)V(3q;6q)\geq V(3q;\infty) + \delta \rho^\eta(\infty) V(d^\eta(\infty);\infty)
\end{equation}
which we will now show for some $\eta$.

For the moonshot $d=6q$, we obtain as value in $t=1$, $V(6q;\infty)=2/\sqrt{3}q$, and in $t=2$, $V(3q; 6q)=\left(3-2/\sqrt{3}\right)q$. The success probability in $t=2$ follows from \Cref{prop:substitutes} for deepening research on a research area with $X=6q$ while being on the boundary distance $d(X=6q)=X/2$. 
For the myopic optimum $d=3q$, we obtain as value in $t=1$, $V(3q;\infty)=3/2q$, and in $t=2$, $V(d^\eta(\infty); \infty)=d^\eta(\infty)(1 - d^\eta(\infty)/(6q))$. The distance and success probability in $t=2$ follow from \Cref{prop:substitutes} for expanding knowledge.
Plugging into \eqref{eq:designer_prefers_fill_up};
\[2/\sqrt{3}q  + \delta \rho^\eta(6q)(3-2/\sqrt{3})q \geq 3/2 q + \delta \rho^\eta(\infty)d^\eta(\infty)(1-d^\eta(\infty)/(6q)).\]
By continuity in $\eta$ and $\delta$, it suffices to show that, for $\delta=1$ and some $\eta>0$ this inequality is strict.
Numerically solving for $d^\eta(\infty),\rho^\eta(\infty),\rho^\eta(6q)$ using $(\eta=0.25,q=1$) verifies strict inequality.\footnote{In this case, $\rho^\eta(6q)=0.8457, \rho^\eta(\infty)=0.7174, d^\eta(\infty)=2.3988$.} As $d^\eta(\infty)$ is linear in $q$ (see \cref{homoq}), $\rho^\eta(\infty)$ is constant in $q$. Linearity of distance and invariance of probability in the moonshot are immediate. Thus, restricting attention to $q=1$ is without loss.


\subsection{Proof of Proposition \ref{prop:deepen_knowledge}} 
\label{sub:proof_of_proposition_prop:deepen_knowledge}


\begin{proof}

We first state a relationship between different types of knowledge for the dynamic setting. We begin by defining a (forward-looking) notion of equivalence. 

\begin{definition}[Forward Equivalence of Knowledge]
Knowledge $\mathcal F_t$ and $\mathcal F'_t$ are considered forward equivalent at time $t$, if---absent a disclosure opportunity at time $t$---the expected future generated payoff is the same.
\end{definition}
Note that the current value of $\mathcal F_t$ and $\mathcal F'_t$ need not be the same, but are inconsequential for the designer's decisions moving on.

\begin{lemma}\label{lem:equivalence}
Knowledge $\mathcal F$ is forward equivalent to $\mathcal F_0$ if it is dense, and there have been no failures in the past.
\end{lemma}

\begin{proof}
If $\mathcal F$ is dense, no researcher  ever deepens knowledge inside areas existing under $\mathcal F$ because $X<\widehat{X}^\eta$. No designer would disclose a question to deepen knowledge, because $X<\widehat{X}^0$. Hence behavior and the generated payoffs are the same.
\end{proof}

\begin{lemma}\label{lem:full_up}
Suppose $\mathcal F_t = \{(0,y(0)),(x_1,y(x_1))\}$ with $x_1\in [4q,\min\{\widetilde{X}^\eta,2\widehat{X}^\eta\}]$ and that there was a failure at $x_1/2$. 
Further, assume the designer can disclose one answer at $t+1$ for free and expects no future disclosure. Then, if $(\eta,\delta)$ is such that $\delta \rho^\eta(\infty)>1/2$, the designer strictly prefers to disclose the answer to $x_1/2$ than to $2 x_1$.
\end{lemma}

\begin{proof}
If the designer discloses the answer to $x_1/2$ her continuation payoff is
\[V(x_1/2;x_1) + \sum_{k=1}^\infty (\delta \rho^\eta(\infty))^k V(d^\eta(\infty);\infty)\]

Alternatively, the designer could disclose the answer to question $2x_1$ thereby starting a new research cycle. Her expected payoff from that action is 
\(V(x_1;\infty) + \delta\rho^\eta(x_1) V(x_1/2;x_1)\)
because after the new research cycle is completed, researchers revert back to trying to find the answer to question $x_1/2$ and fail due to symmetry. 

The designer strictly prefers the disclosure of $x_1/2$ if

\begin{equation}\label{eq:designer_prefers_fill_up}
\begin{split}
    &V(\frac{x_1}{2};x_1) + \sum_{k=1}^\infty (\delta \rho^\eta(\infty))^k V(d^\eta(\infty);\infty) >  V(x_1;\infty) + \delta\rho^\eta(\frac{x_1}{2};x_1) V(\frac{x_1}{2};x_1) \\
    \Leftrightarrow &\sum_{k=1}^\infty (\delta \rho^\eta(\infty))^k > \frac{V(x_1;\infty)- (1-\delta \rho^\eta(x_1/2;x_1)) V(x_1/2;x_1)}{V(d^\eta(\infty);\infty)}.
\end{split}
\end{equation}
Notice that $V(x_1;\infty)/V(d^\eta(\infty);\infty)<V(x_1;\infty)/V(2q;\infty)\leq 1$ because $x_1\geq 4q$. As $\delta \rho^\eta(x_1/2;x_1)<1$ and $V(x_1/2;x_1)>0$ an upper bound for the RHS is $1$. Inequality \eqref{eq:designer_prefers_fill_up} holds if \(\delta \rho^\eta(\infty)/(1-\delta \rho^\eta(\infty))\geq 1\),
which is true if $\rho^\eta(\infty)\delta >1/2$.
\end{proof}
Combining \cref{lem:equivalence,lem:full_up} with the fact that payoffs converge to the ``no future disclosure opportunity'' payoffs as $\lambda \rightarrow 0$ yields the proposition.\footnote{Use the parameters $(\eta=0.25,q=1)$ from the proof of \cref{prop:moonshots}, assume  $\delta\geq \rho^\eta({\infty})$ so the environment is promising, and observe that for $\delta=\rho^\eta({\infty})$ a moonshot of around $\approx 5.3$ is optimal. Local continuity implies a non-knife-edge case.}
\end{proof}

\setlength\bibitemsep{1pt}
\begin{refcontext}[sorting=nyt]
\printbibliography
\end{refcontext}

\newpage
\setcounter{page}{1}
\renewcommand*{\thepage}{S.\arabic{page}}

\renewcommand{\partname}{}
\renewcommand{\thepart}{}
\begin{refsection}
\part{Online Appendix} 
\label{sec:omitted_steps_in_the_proofs}

\section{Different Universe of Questions} 
\label{sub:different_universe_of_questions}

Our baseline model assumes that the universe of questions can be represented on the real line. That is, we implicitly assume an order on questions. In this part, we show that all our results extend to a more general question space.

To begin with, consider our baseline model and fix some knowledge $\mathcal F_m$. As described in \cref{sec:model}, knowledge pins down $\mathcal X_k$\textemdash a set composed of (half-)open intervals: bounded intervals $[x_i,x_{i+1})$ of length $X_i$ each, and two unbounded intervals $(-\infty,x_1)$ and $[x_k,\infty)$ of length $\infty$. All our results survive if knowledge partitions the question space into a set of intervals on the real line (possibly of infinite length).

To see this, consider any set $\widehat{ \mathcal{ X}}_m=\widehat{ \mathcal{ X}}_k \cup \widehat{ \mathcal{ X}}_n$ that contains $k+n$ elements: $k \geq 0$ convex-valued and bounded intervals on $\mathbb{R}$ with Euclidean distance between its upper and lower bound, $X_{i\in \widehat{\mathcal{X}}_k}$, and $n>0$ convex-valued but unbounded intervals on $\mathbb{R}$ of infinite length, $X_{i\in \widehat{\mathcal{X}}_n}=\infty$. For any tuple $(d,X)$ with $X \in \widehat{ \mathcal{ X}}_m$ and $d\in[0, X/2]$ all our definitions and expressions for benefits and cost are well-defined regardless of how $\widehat{ \mathcal{ X}}_m$ was generated.

For any given set $\widehat{ \mathcal{ X }}_m$ generated by some existing knowledge $\mathcal{F}_m$, suppose that the truth-generating process $Y$ is such that the answer to question $x$ characterized by $(d,X)$ is normally distributed with a variance of $\sigma^2(d;X)$.\footnote{Note that the dependence of the variance of the conjecture depends solely on $d$ and $X$. Thus, the truth-generating process has to satisfy a Markov property as the Brownian motion on the real line in our main model. Moreover, note that the specification of the expected value of the answer is not relevant for our results as long as it is well-defined given $\mathcal{F}_m$} Then, all of our results continue to hold. 

Using this formulation, it becomes clear which formal requirements we impose on the set of questions: (i) There are no circular paths in the set of questions $\widehat{\mathcal X}_m$, (ii) the set of questions is piecewise convex-valued, (iii) there is at least one unbounded area. One way to interpret these requirements is to assume a forest network in which the set of nodes represents knowledge and each edge represents an area. We augment this network with (at least) one ``frontier''---a standard Wiener process, and define Brownian bridges over each edge of the network.

\Cref{fig:generalX} depicts different question spaces. While the left panel is our baseline setting, the other two provide alternatives, in the middle there are a number of different directions at which we could expand the frontier starting from the origin. In each direction, the truth would be defined by an independent Wiener process starting at $(0, y(0))$. The right panel shows the limit in which the number of directions becomes a continuum. Still, in each direction there is an independent Wiener process governing the truth over this part of the question space such that circles are excluded. 

\begin{figure}[htb]
    \centering
	\subfloat[Baseline]{
        \includestandalone[width=.28\textwidth]{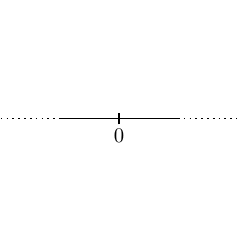}
        }\hfill
    \subfloat[Tree]{
        \includestandalone[width=.28\textwidth]{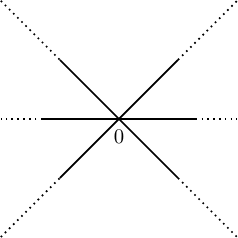}
        }\hfill
    \subfloat[Ocean]{
     \includestandalone[width=.28\textwidth]{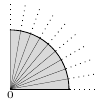}
     }
	\caption{\emph{Different Question Spaces.}}\label{fig:generalX}
\end{figure}

For our static considerations (\cref{sec:the_value_of_increasing_knowledge,sub:selecting_a_research_questions}), all of these models are equivalent because at least one area of infinite length exists at all times allowing the researcher to expand knowledge. Because there are no circles, knowledge partitions the question space into (conditionally) independent segements, just like in the baseline version.

For our dynamic considerations (\cref{sub:research_dynamic}), all models are identical as well as long as we focus on symmetric \emph{pure} strategies. Because researchers ignore previous failures and always pick the same direction, the number of additional directions is not important. If we allow for \emph{mixed} strategies instead, the models differ. For example, when randomizing over the direction from the continuum of directions of the ocean model,  the probability that the researcher draws a direction that a previous generation had selected (and failed at) is zero. 
As we discuss in \cref{sec:final_remarks}, such a model may provide a microfoundation for our assumption of ``non-observed failures.''

We now describe two specific extensions to our baseline setting to illustrate the abstract discussion above.
 
\subsection{General Universe of Questions}

Here, we show a mapping from a model with an $n-$dimensional question space. Suppose that the set of research questions consists of $n$ real lines, $\mathcal{I}=\{I_1, \dots, I_n\}$. In addition, the answers on each of these real lines are determined by a realized path of a one-dimensional standard Brownian motion, such that the truth-generating process is an $n$-dimensional independent Brownian motion $W_z=(W_z^1,...,W_z^n)$.\footnote{Each process starts at an initial point $(0,0)$, has a drift of zero, a variance of one, and independent, normal increments.} Suppose $\mathcal F_{j(i)}^i$ is the finite set of $j(i)$ known realizations of the Brownian path in dimension $i$ and $\mathcal F_k = \cup_{i=1}^n \mathcal F_{j(i)}^i$ is knowledge. As described in \cref{sec:model}, each $\mathcal F_{j(i)}^i$ determines a partition of the domain of $W_z^i$ denoted by $\mathcal X_{j(i)}^i$ with $j(i)+1$ elements. As in the baseline case, the knowledge in dimension $i$ decomposes the dimension-$i$ process into $j(i)-1$ independent Brownian bridges each associated with a length $X^i_l$, $l=\{1,...,j(i)-1\}$ and two independent Brownian motions. Therefore, the union $\mathcal F_k$ determines $k=\left(\sum_{i=1}^n j(i)\right)-n$ independent Brownian bridges of length $X^i_l$ each and $2n$ Brownian motions. By the martingale property of the Brownian motion and the fact that realizations are not directly payoff relevant, the setting is isomorphic to one in which we have $k$ independent \emph{standard} Brownian bridges of length $X_l^i$ each and $2n$ \emph{standard} Brownian motions. Thus, the set $\widehat{\mathcal X}_k=\{X_{l(i)}^i\} \cup \{\infty\}$ is a sufficient statistic to calculate any of the results in the text. However, the set $\widehat{\mathcal X}_k=\{X_{l(i)}^i\} \cup \{\infty\}$ can also be generated with an appropriate realized path of a one-dimensional Brownian motion with a corresponding $\mathcal F_k$. 

\subsection{Seminal Discoveries}
We conclude this part by presenting a model with \emph{seminal discoveries}---discoveries that open new fields of research---that builds on the multidimensional universe of questions described above. For example, Friedrich Miescher's isolation of the ``nuclein'' in 1869 was initially intended to contribute to the study of neutrophils, yet, in addition, it opened up the new and, to a large extent, orthogonal field of DNA biochemistry. 

Formally, consider the following model of the evolution of knowledge. Initially, there is a single field of research $A$ and a single known question-answer pair, $(x_0,y(x_0))=(0,0)$. The set of all questions in field $A$ is known to be one-dimensional and represented by $\mathbb{R}$. The truth is known to be generated by a standard Brownian path $Y$ passing through $(0,0)$. However, with an exogenous probability $p \in [0,1]$ any discovery $(x,y(x))$ is \emph{seminal} and opens a new, independent field of research $B_x$. A seminal discovery is a question-answer pair $(x,y(x))$ that is an element of two independent Brownian paths crossing only at $(x,y(x))$. Thus, upon occurrence, a seminal discovery generates knowledge in multiple dimensions. Because it is a priori unknown whether a discovery is seminal, the payoff from generating knowledge in another dimension is constant in expected terms---it does not influence a researcher's (or designer's) choices. After the seminal discovery, the updated model of truth and knowledge is the one described above with the multi-dimensional universe of questions. As we argued above, that model can, in turn, be mapped into our baseline. The special case with $p=0$ is our baseline model.

It should become clear from our discussion that even the case in which the probability of a seminal discovery depends on the question is qualitatively similar to what we discuss in the baseline model. The quantitative differences in such a model come from the fact that questions which are likely to be a seminal discovery are more attractive to address for all parties involved.


\section[Comparison to Kuhn]{Comparison to \citet{kuhn2012structure}}

In this part, we compare our model and findings in detail to the work of \citet{kuhn2012structure}. We demonstrate which aspects we cover and where we differ from his seminal ideas.

\paragraph{Similarities.} Kuhn himself claims that the \begin{quote}
    paradigm as shared example[s] is the central element of the most novel and least understood aspect of `The Structure of Scientific Revolutions.' \begin{flushright}
        \citep[p. 187]{kuhn2012structure}.
    \end{flushright}
\end{quote} 
Our concept of ``conjectures'' based on existing knowledge offers an economic interpretation of this idea. The paradigm in a specific field is determined by discoveries (that constitute knowledge in our model) and their implications (which follow from the conjectures derived from knowledge). The discoveries and derived conjectures in our model provide comprehensive information on how society and researchers approach questions.

Furthermore, we believe that our framework and the reoccuring phases of expanding research in our model (see \Cref{prop:laissez-fair}) closely resemble an economic approach to \citet{kuhn2012structure}'s idea of \emph{normal science} (described in \citet{kuhn2012structure}, Chapter 3). Normal science in the Kuhnian sense consists of researchers solving puzzles in the context of a given paradigm. Similarly, researchers in our model build on existing knowledge (and the truth following a Brownian motion) to form conjectures about the location of answers and search for answers where they expect them to be. 
Whenever a researcher finds an answer, she finds it close to where it was expected to be. 


Our analysis of this model then gives rise to a formalized economic theory that explains ``how little [researchers] aim to produce major novelties, conceptual or phenomenal'' \citep[p. 35]{kuhn2012structure}. 


\paragraph{Differences.} Much in the spirit of \citet{kuhn2012structure}, there are times at which normal science in our model fails to advance knowledge: The Brownian motion eventually takes an unexpected turn, and researchers will fail to recognize them when expanding knowledge step-by-step building on their conjectures. \citet{kuhn2012structure}'s somewhat radical idea then is that anomalies \citep[Chapter 6]{kuhn2012structure} pile up to a crisis \citep[Chapter 7]{kuhn2012structure} in which normal science desperately tries to connect seemingly contradictory information to an old model of the world. At least in our baseline model, such chaos is absent. Instead, while research produces discoveries during the step-by-step expansion of knowledge, conjectures shift only gradually because the discoveries realize within the search intervals chosen by the researcher, and therefore, close to where they were expected. This gradual revision of conjectures appears closer to \citet{toulmin1970does}'s evolutionary model of science, with constant adaptation and revision of theories. 
When a researcher fails finding an answer building on her conjecture, science gets stuck as researchers keep repeating the same mistakes applying the old, misleading conjecture. Only some exogenous shock (for example, a moonshot or a serendipitous discovery) can take them out of their misery by providing new guidance. Such exogenous shocks may be closer to what \citet{kuhn2012structure} has in mind as it follows a period of little progress and sparks a sudden appearance of highly productive research. Thus, while \citet{kuhn2012structure} depicts ``revolutions'' as settings where researchers go out of their way and wildly experiment with no discipline to uncover the new paradigm, such phases in our setting are either directed (through moonshots) or happen by chance when normal science seizes to advance knowledge.




\paragraph{Summary.} Our model can capture many of the observations \citet{kuhn2012structure} makes. Yet, there are important differences in the mechanics: While \citet{kuhn2012structure} diagnoses that ``[w]ithout commitment to a paradigm there could be no normal science'' \citep[p.100]{kuhn2012structure}, our paradigms evolve as researchers go along. While our findings are in line with \citet{kuhn2012structure}'s statement that ``surprises, anomalies, or crises \ldots are just the signposts that point the way to extraordinary science,” (p.101), our framework does not describe that researchers respond to crisis by ``searching at random, trying experiments just to see what will happen'' \citep[p.87]{kuhn2012structure}. 

In our world, it is not that there is a missing link in the set of problems to cover, but a missing connection to the problems down the line that needs to be discovered. 
In \citet{kuhn2012structure}'s world, crisis is the (endogenous) driver behind radical change. He claims that:

\begin{quote}
Confronted with anomaly or with crisis, scientists take a different attitude toward existing paradigms, and the nature of their research changes accordingly.\begin{flushright}
    \citet[p. 90f]{kuhn2012structure}.
\end{flushright}
\end{quote} 

Because the researcher is freed from paradigmatic discipline when in crises, his mind evolves more freely, yet the resolution is then modeled as the ``stroke of genius'' or, as he puts it:

\begin{quote}
   More often no such structure is consciously seen in advance. Instead, the new paradigm, or a sufficient hint to permit later articulation, emerges all at once, sometimes in the middle of the night, in the mind of a man deeply immersed in crisis. What the nature of that final stage is—-how an individual invents (or finds he has invented) a new way of giving order to data now all assembled—-must here remain inscrutable and may be permanently so. \begin{flushright}
    \citep[p. 89f]{kuhn2012structure}.
   \end{flushright}
\end{quote}

While the consequences of a stroke of genius are similar in our model to what we believe \citet{kuhn2012structure} has in mind, we do not model the crisis-plagued researcher. Partially, that is because it remains also opaque within \citet{kuhn2012structure} how such thinking comes about. Instead, we ask whether and how well-designed funding institutions (absent in \citet{kuhn2012structure}) can help to start research cycles that reduce the risk of crisis and the need for a genius. Here, we connect to \citet{merton1957priorities} or \citet{PARTHA1994487} and take the general freedom of scientists as given, but also acknowledge that they respond to incentives \citep[see, e.g.,][]{myers2020elasticity}, a notion completely absent in \citet{kuhn2012structure}. 

\section{The Cost of Research and Proof of Lemma \ref{lem:cost}} 
\label{sec:the_cost_of_research}

In this section, we provide a detailed derivation of the cost of research. \Cref{lem:cost} is a corollary to the results we obtain. The cost implies an endogenous measure of the productivity of research. We model research as sampling a set of candidate answers to question $x$ with the goal of discovering the actual answer, $y(x)$.

Formally, we assume that, conditional on a question $x$, the sampling decision consists of selecting an interval $[a,b] \in \mathbb{R}$. If the true answer lies inside the chosen interval, such that $y(x) \in [a,b]$, research succeeds and a discovery is made. If $y(x) \notin [a,b]$, research fails and no discovery is made. Thus, the choice of the research interval entails an ex-ante probability of successful research. Restricting the sampling decision to a single interval $[a,b]$ comes without loss for our purposes, as conjectures $G_x(y|\mathcal F_k)$ follow a normal distribution.

We now characterize the cost of research in terms of the three variables of interest: the research area, $X$, the novelty of the question, $d$, and the expected output, $\rho$. 

We begin by defining a prediction interval and characterize it based on the conjecture $G_x(y|\mathcal F_k)$.

\begin{definition}[Prediction Interval]
The prediction interval $\alpha(x,\rho)$ is the shortest interval $[a,b] \subseteq \mathbb{R}$ such that the answer to question $x$ is in the interval $[a,b]$ with probability $\rho$.
\end{definition}


\begin{figure}\centering
\includestandalone[width=.45\textwidth]{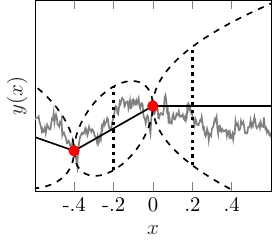}
\caption{\emph{Cost of research and interference.} \footnotesize{The dotted vertical lines represent the 95\% prediction intervals for the answers to questions $x=-0.2$ and $x'=0.2$, assuming the answer to questions $0$ and $-0.4$ are known. 
}}\label{fig:pred_int}
\end{figure}

\begin{proposition}\label{lem:pred_interval}
Suppose $\alpha(x,\rho)$ is the prediction interval for probability $\rho$ and question $x$ when answer $y(x)$ is normally distributed with mean $\mu$ and standard deviation $\sigma$. Then, any prediction interval has the following two features:
\begin{enumerate}
	\item The interval is centered around $\mu$.
	\item The length of the prediction interval is $2^{3/2} \operatorname{erf}^{-1}(\rho) \sigma$, where $\operatorname{erf}^{-1}$ is the inverse of the Gaussian error function.
\end{enumerate} 
\end{proposition}
\begin{proof}
The normal distribution is symmetric around the mean with a density decreasing in both directions starting from the mean. It follows directly that the smallest interval that contains the realization with a particular likelihood is centered around the mean. 

Take an interval $[z_l,z_r]$ of length $Z<\infty$ that is symmetric around the mean $\mu$ and let it be such that it contains a total mass of $\rho<1$ in the interval. Then, a probability mass of $(1-\rho)/2$ lies to the left of the interval by symmetry of the normal distribution. Moreover, the left bound $z_l$ of the interval has (by symmetry of the interval around the mean $\mu$) a distance $\mu-Z/2$ from the mean. From the properties of the normal distribution, \[\Phi(z_l) = 1/2 \left(1+\operatorname{erf}\left(\frac{z_l-\mu}{\sigma \sqrt{2}}\right)\right)=1/2 \left(1+\operatorname{erf}\left(\frac{-Z/2}{\sigma \sqrt{2}}\right)\right).\]
Solving, using the symmetry of $\operatorname{erf}$, yields
\[1/2 \left(1-\operatorname{erf}\left(\frac{Z}{\sigma 2^{3/2}}\right)\right)= \frac{1- \rho}{2} \: \Leftrightarrow \: \operatorname{erf}\left(\frac{Z}{\sigma 2^{3/2}}\right) = \rho 
\Leftrightarrow  Z= 2^{3/2} \operatorname{erf}^{-1}(\rho) \sigma.\qedhere\]
\end{proof}

\Cref{fig:pred_int} indicates that the prediction interval depends on the location of the question. Two questions with the same distance from existing knowledge (that is, distance from question $x=0$) have different 95\% prediction intervals depending on whether research deepens knowledge or expands it. That difference translates into different costs. 

\cref{lem:pred_interval} implies that if the cost function is homogeneous of any degree in interval length $(b-a)$, we can represent it with an alternative cost function proportional to $c(\rho,d,X)$ that is multiplicatively separable in $(d,X)$ and $\rho$ without having to keep track of the exact location of the search interval $[a,b]$, which proves to be convenient.

It also implies that, fixing $\rho$, the changes in the cost with respect to distance $d$ and area length $X$ vary in their effect on $\sigma(d;X)$ only. Similarly, holding distance and area length constant, changes in $\rho$ translate into cost changes according to a function of $\operatorname{erf}^{-1}(\rho)$\textemdash a convex increasing function. 

\cref{lem:pred_interval} intuitively links the cost of research effort to the probability of a discovery. Because the inverse error function is increasing and convex, the cost of finding an answer with probability $\rho$ is increasing and convex in $\rho$. 

In the paper we assume the cost to be proportional to $(a-b)^2$. As should be clear from \cref{lem:pred_interval}, the quadratic formulation is for convenience only. What matters for our results qualitatively is that the cost is (i) homogeneous, (ii) increasing, and (iii) convex in the sampling interval $(a-b)$. Under the quadratic assumption, the cost function is characterized by a simple corollary to \cref{lem:pred_interval}: \cref{lem:cost}. 

\begin{samepage}
\begin{corollary}
For knowledge $\mathcal{F}_k$, probability $\rho$, and question $x$, the minimal cost of obtaining an answer to question $x$ with probability $\rho$ is proportional to
\(c(\rho,d;X)= \tilde{c}(\rho) \sigma^2(d;X).\)
\end{corollary}
\end{samepage}


\section[Properties of c(rho)]{Properties of $\tilde{c}$} 
\label{sec:Notation}

\paragraph{Summary.} The function $\tilde{c}(\rho)$ is convex and increasing on $[0,1)$ with $\tilde{c}(0)=0$ and $\lim_{\rho \rightarrow 1}\tilde{c}(\rho)=\infty$.\footnote{Due to this limit and the researcher's ability to choose $\rho=1$, we augment the support of the cost function to include $\rho=1$ with $\tilde{c}(1)=\infty$. However, the optimal $\rho$ is always strictly interior unless the cost parameter $\eta$ is chosen to be zero in which case we assume that $\eta \tilde{c}(\rho=1)=0$.}
The derivative
\(\tilde{c}_\rho(\rho)=\sqrt{\pi} \operatorname{erf}^{-1}(\rho) e^{\tilde{c}(\rho)}\)
is increasing and convex with the same limits.

We use that, for $\rho \in (0,1)$, $\tilde{c}(\rho)$ has a convex and increasing elasticity bounded below by $2$ and unbounded above. Its derivative $\tilde{c}_\rho(\rho)$ has an increasing elasticity bounded below by $1$ and unbounded above. We want to emphasize that these properties are not special to our quadratic cost assumption. To the contrary, $\operatorname{erf}^{-1}(x)^k$ for any $k\geq 2$ admits similar properties with only the lower bounds changing. The following properties are invoked in the proofs:

\noindent \begin{minipage}{0.45\linewidth}
\begin{align*}
\rho \frac{\tilde{c}_{ \rho}(\rho)}{\tilde{c}(\rho)} &\in (2,\infty) \text{ and increasing,} \\
\rho \frac{\tilde{c}_{\rho \rho}(\rho)}{\tilde{c}_\rho(\rho)} &\in (1,\infty) \text{ and increasing,}
\end{align*}
\end{minipage}
\begin{minipage}{0.45\linewidth}
\begin{align*}
\rho \tilde{c}_\rho(\rho) -\tilde{c}(\rho) &\in (0,\infty) \text{ and increasing,}\\
\tilde{c}_\rho^{-1}(x)&= \operatorname{erf}\left(\sqrt{\frac{W(2x^2/\pi)}{2}}\right).
\end{align*}
\end{minipage}

with $W(\cdot)$ the principal branch of the Lambert W function. We formally prove the properties that do not directly follow from the definition of the inverse of the error function below.

\subsection[Properties of c(rho)]{Proofs of properties of $\tilde{c}(\rho)$}

Here, we provide the formal proofs. To simplify notation, we suppress the argument $\rho$ and denote the inverse error function by $\iota:=\operatorname{erf}^{-1}(\rho)$.
\begin{lemma}\label{lem:inv_err_properties}
	The derivatives of the inverse error function satisfy (i) $\frac{d}{d\rho} \iota  = \frac{1}{2}\sqrt{\pi}e^{\left(\iota^2\right)}$, (ii) $\frac{d^2}{d\rho^2}\iota  = 2 \iota \iota'^2$, and (iii) $\frac{d^3}{d\rho^3} \iota  = 2\iota'^3\left(1+4\iota^2\right)$.
\end{lemma}
\begin{proof}
	See \citet{Dominici2008}.
\end{proof}

\begin{lemma}\label{lem:helpful_limit} (i) $\lim_{\rho\rightarrow 0} \rho \frac{\iota'}{\iota}=1$, (ii) $\lim_{\rho\rightarrow 1} \rho \frac{\iota'}{\iota}=\infty$, (iii) $\lim_{\rho\rightarrow 0} \frac{d}{d\rho}\left(\rho \frac{\iota'}{\iota} \right) = 0$, and (iv) $\lim_{\rho \rightarrow 0}\frac{d^2}{d\rho^2}\left(\rho \frac{\iota'}{\iota} \right) = \frac{\pi}{3}$.
\end{lemma}
	\begin{proof}
	We will make use of L'H\^opital's rule and the properties from \Cref{lem:inv_err_properties}.

	\noindent (i) follows from 
	\begin{align*}
		 \lim_{\rho\downarrow 0} \rho \frac{\iota'}{\iota} &= \lim_{\rho\downarrow 0} \frac{\iota'+\rho  \iota''}{\iota'} =\lim_{\rho\downarrow 0} \frac{\iota'+2\rho \iota \iota'^2}{\iota'} =\lim_{\rho\downarrow 0} (1+\rho \iota \iota') =1.
	\end{align*}
	(ii) follows from
		\begin{align*}
		 \lim_{\rho\uparrow 1} \rho \frac{\iota'}{\iota} &= \lim_{\rho\uparrow 1} \frac{\iota'+\rho  \iota''}{\iota'} =\lim_{\rho\uparrow 1} \frac{\iota'+2\rho \iota \iota'^2}{\iota'} =\lim_{\rho\uparrow 1} (1+2\rho \iota \iota') =\infty.
	\end{align*}
	(iii) follows from 
		\begin{align*}
		\lim_{\rho\rightarrow 0} \frac{d}{d\rho}\left(\rho \frac{\iota'}{\iota} \right) &= \lim_{\rho\rightarrow 0}\frac{\iota'}{\iota}\left(1-\rho \frac{\iota'}{\iota}\right) +\lim_{\rho\rightarrow 0} \rho \frac{\iota''}{\iota} = \underbrace{\lim_{\rho\rightarrow 0}\iota'}_{=\sqrt{\pi}/2} \lim_{\rho\rightarrow 0}\frac{\iota-\rho \iota'}{\iota^2} +\underbrace{\lim_{\rho\rightarrow 0} 2\rho \iota'^2}_{ = 0} \\ & =-\lim_{\rho\rightarrow 0}\frac{\sqrt{\pi}}{2}  \frac{\rho \iota''}{2 \iota \iota'} = - \lim_{\rho\rightarrow 0}\frac{\sqrt{\pi}}{2}  \frac{\rho \iota (\iota')^2}{2 \iota \iota'} =-\lim_{\rho\rightarrow 0} \frac{\sqrt{\pi}}{2}  \rho \iota' = 0.
	\end{align*}
	(iv) follows from\footnote{To arrive at the first line let $\lambda:=\iota'/\iota$ and observe that $(\rho \lambda)''= (\lambda+ \rho \lambda')' = 2 \lambda' + \rho \lambda''$ and $\lambda' = 2(\iota')^2 - \lambda^2$ which implies $\lambda''= 4 \iota' \iota'' - 2 \lambda \lambda'$.}
	\begin{align*}
		\lim_{\rho \rightarrow 0}\frac{d^2}{d\rho^2}\left(\rho \frac{\iota'}{\iota} \right) &=\lim_{\rho \rightarrow 0} 2\frac{\iota'' \iota-\iota'^2}{\iota^2}\left(1-\rho\frac{\iota'}{\iota}\right) +\underbrace{\lim_{\rho \rightarrow 0} 4 \rho \underbrace{\iota' \iota''}_{=2(\iota')^3 \iota}}_{=0} \\ & =\lim_{\rho \rightarrow 0} 2\frac{\iota'' \iota-\iota'^2}{\iota^2}\left(1-\rho\frac{\iota'}{\iota}\right)=\lim_{\rho \rightarrow 0} 2\frac{\iota'' \iota}{\iota^2}\left(1-\rho\frac{\iota'}{\iota}\right) -2\lim_{\rho \rightarrow 0}\frac{\iota'^2}{\iota^2}\left(1-\rho\frac{\iota'}{\iota}\right) \\ & =\underbrace{\lim_{\rho \rightarrow 0} 4 \iota'^2 \left(1-\rho\frac{\iota'}{\iota}\right)}_{=0} -2\lim_{\rho \rightarrow 0}\frac{\iota'^2}{\iota^2}\left(1-\rho\frac{\iota'}{\iota}\right) =-2\lim_{\rho \rightarrow 0}\left(\rho\frac{\iota'}{\iota}\right)^2 \frac{\iota-\rho\iota'}{\rho^2\iota} \\ &=2\lim_{\rho \rightarrow 0} \frac{\rho\iota''}{2\rho \iota+\rho^2\iota'} =4\lim_{\rho \rightarrow 0} \frac{\iota'^2}{2+\rho\frac{\iota'}{\iota}} =\frac{4}{3}\lim_{\rho \rightarrow 0} \iota'^2= \frac{\pi}{3}. \quad \qedhere
	\end{align*}
	\end{proof}
	
\begin{lemma}\label{lem:c_properties}
The following hold: (i) For all $\rho\in(0,1)$, $\frac{d}{d\rho}\left(\rho\tilde{c}_\rho(\rho)-\tilde{c}(\rho)\right)>0$. (ii) For all $\rho\in(0,1)$, $\rho\tilde{c}_\rho(\rho)-\tilde{c}(\rho)>0$. (iii) $\lim_{\rho \rightarrow 0} \rho \frac{\tilde{c}_\rho(\rho)}{\tilde{c}(\rho)}=2$. (iv) $\lim_{\rho \rightarrow 1} \rho \frac{\tilde{c}_\rho(\rho)}{\tilde{c}(\rho)}=\infty$.
\end{lemma}
\begin{proof}
	(i) holds because
	\(\frac{d}{d\rho}\left(\rho\tilde{c}_\rho(\rho)-\tilde{c}(\rho)\right)= \rho \tilde{c}_{\rho \rho}(\rho)>0 \)
	by convexity of the inverse error function.	(ii) holds due to (i) and $\left(\rho\tilde{c}_\rho(\rho)-\tilde{c}(\rho)\right)|_{\rho=0}=0$. 	(iii) holds as the elasticity is equal to $2\rho \frac{\iota'}{\iota}$ and (i) in \Cref{lem:helpful_limit}. 	(iv) holds by the same observations and (ii) in \Cref{lem:helpful_limit}.
\end{proof}

\begin{lemma}\label{lem:elasticity_c}
	The elasticity of $\tilde{c}(\rho)$, $\rho \frac{\tilde{c}_\rho(\rho)}{\tilde{c}(\rho)}$, is increasing in $\rho$.
\end{lemma}
\begin{proof}
	Recall that $\rho \frac{\tilde{c}_\rho(\rho)}{\tilde{c}(\rho)}=2 \rho \frac{\iota'}{\iota}$. Therefore, it is sufficient to prove that the inverse error function has an increasing elasticity.

	Note that 
	\( \frac{d}{d\rho}\left(\rho \frac{\iota'}{\iota} \right) =\frac{\iota'}{\iota} + \rho\frac{\iota''\iota-\iota'^2}{\iota^2}. \)	From \Cref{lem:helpful_limit} know that 
	\begin{align*}
		\lim_{\rho\rightarrow 0} \frac{d}{d\rho}\left(\rho \frac{\iota'}{\iota} \right) = 0 \qquad 
		\lim_{\rho \rightarrow 0}\frac{d^2}{d\rho^2}\left(\rho \frac{\iota'}{\iota} \right) = \frac{\pi}{3}. 
	\end{align*}
	
	Thus, there exists an $\varepsilon>0$ such that the elasticity is increasing for $\rho \in (0,\varepsilon)$. To show that it is increasing for all $\rho\in(0,1)$ suppose --toward a contradiction-- that the derivative of the elasticity crosses 0. In this case, it has to hold that 
	\( \frac{\iota''\iota-\iota'^2}{\iota^2} =-\frac{\iota'}{\rho \iota}.  \)

		Consider the second derivative of the elasticity at such a critical point
		\begin{align*}
			\frac{d^2}{d\rho^2}\left(\rho \frac{\iota'}{\iota} \right)|_{\frac{d}{d\rho}\left(\rho \frac{\iota'}{\iota} \right)=0} &= 2\frac{\iota''\iota-\iota'^2}{\iota^2}\left(1-\rho \frac{\iota'}{\iota}\right)+\rho\frac{\iota'''\iota-\iota''\iota'}{\iota^2} \\
			&=-2\frac{\iota'}{\iota \rho}\left(1-\rho \frac{\iota'}{\iota}\right)+\rho\frac{\iota'''\iota-\iota''\iota'}{\iota^2}=2\frac{\iota'}{\iota \rho}\left(\rho \frac{\iota'}{\iota}-1\right)+2\rho\frac{\iota'^3}{\iota}4\iota^2 >0
		\end{align*}
		where the last inequality follows because the elasticity is weakly greater than one and all other terms are positive.

		Thus, any critical point must be a minimum. However, the elasticity is continuous and increasing at $\rho\in(0,\varepsilon)$. Thus, there is no interior maximum and the elasticity is increasing throughout.
		\end{proof}
		\begin{lemma}
			The elasticity of $\tilde{c}_\rho(\rho)$, $\rho \frac{\tilde{c}_{\rho\rho}(\rho)}{\tilde{c}_\rho(\rho)}$, is increasing in $\rho$.
		\end{lemma}
		\begin{proof}
			Note that the elasticity of $\tilde{c}_\rho (\rho)$ is equal to 
			\begin{align*}
				\rho \frac{\tilde{c}_{\rho\rho}(\rho)}{\tilde{c}_\rho(\rho)}&=\rho \frac{\frac{d}{d\rho}\left( 2 \iota \iota' \right)}{2 \iota \iota'}\\
				&= \rho \frac{2 \iota \iota'' + 2\iota'^2}{2 \iota \iota'}\\
				&= \rho \frac{\iota'}{\iota}\left(2 \iota^2 +1 \right),
			\end{align*}
			where the last equality follows by replacing $\iota''=2 \iota \iota'^2$ from \Cref{lem:inv_err_properties} and factoring out $\iota'^2$. The derivative of this elasticity is 
			\begin{align*}
				\frac{d}{d\rho} \left(\rho \frac{\tilde{c}_{\rho\rho}(\rho)}{\tilde{c}_\rho(\rho)} \right) = \frac{d}{d\rho}\left(\rho \frac{\iota'}{\iota} \right)\left(2 \iota^2 +1 \right) + \frac{d}{d\rho}\left(2 \iota^2 +1 \right) \rho \frac{\iota'}{\iota}.
			\end{align*}
			Note that the second term is unambiguously positive as $\iota'>0$ and $\iota>0$. The sign of the first term is determined by the sign of $ \frac{d}{d\rho}\left(\rho \frac{\iota'}{\iota} \right)$: the derivative of the inverse error function elasticity. It is
			\begin{align*}
				\frac{d}{d\rho}\left(\rho\frac{\iota'}{\iota}\right)&=\frac{\iota''}{\iota'}+\rho \frac{\iota'''\iota'-\iota''^2}{\iota'^2}=\frac{\iota''}{\iota'}+2\rho \iota''^2 (1+2\iota(2\iota-1)).
			\end{align*}
			We know that $\iota''>0$ and $\iota'>0$. Thus, we only need to show that $ 1+2\iota(2\iota-1)>0$. Note that this is a convex function of $\rho$ with a minimum at $\iota \iota'=\frac{1}{4}$ which is solved by $\rho=\operatorname{erf}\left(\sqrt{\frac{W\left(\frac{1}{2\pi}\right)}{2}}\right)\approx 0.29$ where $W$ denotes the principal branch of the Lambert-W function. Evaluating $ 1+2\iota(2\iota-1)$ at this minimum yields 
			\begin{align*}
				1+ \left(\sqrt{2 W\left(\frac{1}{2\pi}\right)}-1\right)\sqrt{2 W\left(\frac{1}{2\pi}\right)}\approx 0.75.
			\end{align*}
			Hence, we can conclude that $ \frac{d}{d\rho}\left(\rho \frac{\iota'}{\iota} \right)$ is positive and the result follows. 
		\end{proof}

\section{Comparative Statics of Expanding Knowledge}
\label{sec:cs_expanding}
	Recall that the optimal distance $d^\eta(\infty)$ and the optimal success probability $\rho^\eta(\infty)$ when expanding knowledge are implicitly defined by the system of first-order conditions. The comparative statics then follow from applying the implicit function theorem. In particular, we obtain

	\begin{align}
		\left(\begin{matrix} \frac{d}{d\eta} d^\eta(\infty) \\ \frac{d}{d\eta} \rho^\eta(\infty) \end{matrix}\right) = - \frac{1}{det(\mathcal{H})} \left(\begin{matrix} f_{d\eta} f_{\rho \rho} - f_{\rho \eta}f_{d\rho} \\ f_{\rho \eta}f_{dd} - f_{d\eta}f_{d\rho} \end{matrix}\right)
	\end{align}
	where we use the shorthand  $f=u_R(d,\rho;X)$ for the researcher's value from expanding knowledge with distance $d$ and success probability $\rho$.  We obtain 
	\begin{align*}
		f_{d\eta} &= - \tilde{c}(\rho) \\
		f_{\rho \eta} &= - \tilde{c}_\rho(\rho) d \\
		f_{dd} &= \rho V_{dd}(d;\infty) \\ 
		f_{\rho \rho} &= - \eta \tilde{c}_{\rho \rho}(\rho)d \\
		f_{d\rho} &= V_d(d;\infty) - \eta \tilde{c}_\rho(\rho).
	\end{align*}
	Suppressing the point of evaluation and plugging in, the comparative statics yields at the optimal distance $d^\eta(\infty)$ and success probability $\rho^\eta(\infty)$
	\begin{align*}
		\frac{d}{d\eta} d^\eta(\infty) &= - \frac{1}{det(\mathcal{H})} \left( \tilde{c}(\rho) \eta \tilde{c}_{\rho \rho}(\rho)d + \tilde{c}_\rho(\rho) d (V_d(d;\infty) - \eta \tilde{c}_\rho(\rho))  \right) \\ 
		&= - \frac{\eta d}{det(\mathcal{H})} \left(\tilde{c}(\rho) \tilde{c}_{\rho \rho}(\rho) + \tilde{c}_\rho(\rho)(\tilde{c}(\rho)/\rho - \tilde{c}_\rho (\rho))  \right) \\
		&< 0,
	\end{align*}
	where the first equality follows from the first-order condition with respect to $d$ and factoring out $\eta d$. The inequality follows from the fact that the determinant of the Hessian is positive and that the elasticity of $\tilde{c}(\rho)$ is increasing in $\rho$ (\Cref{lem:elasticity_c}).

	For the optimal success probability, we obtain at the optimal distance $d^\eta(\infty)$ and success probability $\rho^\eta(\infty)$
	\begin{align*}
		\frac{d}{d\eta} \rho^\eta(\infty) &= - \frac{1}{det(\mathcal{H})} \left(- \tilde{c}_\rho(\rho) d\rho V_{dd}(d;\infty) +  \tilde{c}(\rho)(V_d(d;\infty) - \eta \tilde{c}_\rho(\rho)) \right) \\ 
		&= - \frac{d/(3q) \tilde{c}(\rho)}{det(\mathcal{H})} \left( \frac{\tilde{c}_\rho(\rho)\rho}{\tilde{c}(\rho)} + \frac{V_d(d;\infty)-V(d;\infty)/d}{d/(3q)} \right) \\ 
		&< 0,
	\end{align*}
	where the equality follows from the first-order condition with respect to $\rho$, plugging in the expression for $V_{dd}(d;\infty)$ (noticing that $d^\eta(\infty)<4q$) and factoring out $d/(3q) \tilde{c}(\rho)$. The inequality follows from the fact that the determinant of the Hessian is positive, that the elasticity of $\tilde{c}(\rho)$ increases in $\rho$ and that it is greater than two (\Cref{lem:c_properties}), and that the term in parentheses is equal to negative one half after plugging in.  

	Hence, both the optimal distance and the optimal success probability decrease in the cost parameter $\eta$ when expanding knowledge.

\section{Omitted Proofs} 
\label{sub:omitted_proofs}

\begin{lemma}\label{lem:inequalityexpandinglarger4}
	$\frac{\partial V(d;\infty|d> 4q)}{\partial d}<0.$
	\end{lemma}
	\begin{proof}
	\[\frac{\partial V(d;\infty|d> 4q)}{\partial d}=-\frac{d}{3q}+1 + \sqrt{\frac{d-4q}{d}} \frac{d-q}{3q}\]
	Letting $\tau:=d/q(>4$ by assumption) the statement is negative if
	\(
	\frac{3-\tau}{3} + \sqrt{\frac{\tau-4}{\tau}} \frac{\tau-1}{3}<0.
	\)
	The left-hand side is increasing in $\tau$ and converges to $0$ as $\tau \rightarrow \infty$. 
	\end{proof}
	
	\begin{lemma}\label{sub:addendum_to_sub:proof_of_cor_deep_optimal}
	$V_d(d;X)>0$ if $d\in [0,X-4q]$ and $X\in(4q,6q]$.
	\end{lemma}

	\begin{proof}
Note that for $X\in (4q,6q]$ and $d \in [0,X-4q]$, 
\begin{align*}
	V_{d} & = \frac{1}{3q}\left(X-2d -(X-d-q)\sqrt{\frac{X-d-4q}{X-d}} \right).
\end{align*}
Assume towards a contradiction that there is a feasible combination of $d$ and $X$ such that $V_d(d;X)\leq 0$. Then, the following inequality must hold 
\begin{align*}
	\frac{X-2d}{X-d-q} \leq \sqrt{\frac{X-d-4q}{X-d}}.
\end{align*}
Observe that this inequality cannot hold at the bounds $d=0$ and $d=X-4q$: If $d=0$, then the left-hand side is strictly greater than one while the right-hand side is strictly less than one. If $d=X-4q$, then the left-hand side reduces to $(8q-X)/(3q)$ which is strictly positive as $X\leq 6q$, while the right-hand side is equal to zero.

Hence, if the inequality is ever satisfied for some feasible $(X,d)$, then by continuity and the intermediate value theorem, there must be a feasible $(\underline{X},\underline{d})$ combination such that $V_d(\underline{d};\underline{X})=0$. Thus, for a feasible $\underline{d}$, there must be a solution $\underline{X}$ to the quadratic equation (in $\underline{X}$)
\begin{align*}
	\frac{X-2\underline{d}}{X-\underline{d}-q} = \sqrt{\frac{X-\underline{d}-4q}{X-\underline{d}}}. 
\end{align*}
The solution to this equation is 
\begin{align*}
	\underline{X}_{1,2} = \frac{1}{4} \left(5 \underline{d} + 3q \pm (\underline{d}-q)\sqrt{\frac{\underline{d}+5q}{\underline{d}-3q}} \right).
\end{align*}
However, no feasible solution exists, as $\underline{d}\leq 2q$ (due to the upper bound $X-4q$ on $\underline{d}$ and the upper bound $6q$ on $X$), leading to no solution for $\underline{X}$ in the real domain. A contradiction.
\end{proof}

\begin{lemma}\label{lem:partialValuetoareanegative}
$V_X(d^0(X);X)<0$ if $X\geq 4q$ and $d\in [0,X-4q]$.
\end{lemma}
\begin{proof}
Observe that for any $X\geq 4q$ and $d \leq X-4q$
\begin{align*}
	V_{Xd} &= \frac{1}{24q}\left(8 - 3 \sqrt{\frac{X-d}{X-d-4q}}-(5(X-d)+4q)\frac{\sqrt{X-d-4q}}{(X-d)^{3/2}} \right).
\end{align*}
Denote $a:=X-d$, this is an increasing function in $a$ as \(\frac{dV_{Xd}}{da}=\frac{4q^2}{a^{5/2}(a-4q)^{3/2}}>0.\)
Hence, the highest value of $V_{Xd}$ is attained for $a\rightarrow \infty$ and
\begin{align*}
	\lim_{a\rightarrow \infty} \frac{1}{24q}\left(8 - 3 \underbrace{\sqrt{\frac{a}{a-4q}}}_{\rightarrow 1}-5\underbrace{\frac{a\sqrt{a-4q}}{a^{3/2}}}_{\rightarrow 1}-4q\underbrace{\frac{\sqrt{a-4q}}{a^{3/2}}}_{\rightarrow 0}\right)=0.
\end{align*}
It follows that the $V_{Xd}$ converges to zero from below implying that $V_{Xd}<0$. Thus, $V_{X}(d^0(X),X)<V_X(d=0,X)$ and we obtain
\begin{align*}
	V_X(d,X|&d\leq 4q,X-d\geq4q)=\frac{d+(X-d-q)\sqrt{\frac{X-d-4q}{X-d}} - (X-q)\sqrt{\frac{X-4q}{X}}}{3q}\\
	&<V_X(d=0,X|d\leq 4q,X-d\geq 4q)=\frac{(X-q)\sqrt{\frac{X-4q}{X}} - (X-q)\sqrt{\frac{X-4q}{X}}}{3q}=0. \quad \qedhere
\end{align*}
\end{proof}

\begin{lemma}\label{lem:secondtotalderivative}
If $X \in (4q,8q)$, $\mathrm{d}^2V(X/2,X)/\mathrm{d}X^2<0$ and $\mathrm{d}^2 V(d^0(X),X)/(\mathrm{d}X)^2>0$. 
\end{lemma}
\begin{proof}

Considering the boundary solution we obtain

\noindent \begin{minipage}{0.45\linewidth}
\begin{align*}
	\frac{d^2V(X/2,X)}{dX^2}&=-\frac{X^2-2qX-2q^2}{3qX^{3/2}\sqrt{X-4q}} + \frac{1}{6q}
\end{align*}
\end{minipage}
\begin{minipage}{0.45\linewidth}
\begin{align*}
\frac{d^3V(X/2,X)}{dX^3}&=\frac{4q^2}{X^{5/2}(X-4q)^{3/2}}>0
\end{align*}
\end{minipage}

implying that $\frac{d^2V(X/2,X)}{dX^2}\leq\frac{d^2V(4q,8q)}{dX^2}$ with
\begin{align*}
	\frac{d^2V(4q,8q)}{dX^2}=-\frac{64q^2-16q^2-2q^2}{3q8^{3/2}q^{3/2} 2q^{1/2}} + \frac{1}{6q}= -\frac{46q^2}{96 \sqrt{2} q^3} + \frac{1}{6q}=\frac{8-23/\sqrt{2}}{48q}<0.
\end{align*}

Next, consider the value of any interior solution and apply the envelope and implicit function theorem to obtain
\begin{align*}
	\frac{d V(d^0(X),X)}{dX}&=V_X + d'(X)\underbrace{V_d}_{=0 \text{ by optimality of }d}=V_X\\
	\frac{d^2V(d^0(X),X)}{dX^2}&= V_{XX} +d'(X)V_{dX} +d'(X)\underbrace{(V_{Xd}+V_{dd}d'(X))}_{=0 \text{ by IFT on FOC}}+ d''(X)\underbrace{V_d}_{=0 \text{ by optimality}}\\
	&=V_{XX}(d^0(X),X)+d'(X)V_{dX}=V_{XX}(d^0(X),X)\underbrace{-\frac{V_{dX}^2}{V_{dd}}}_{>0 \text{ as }V_{dd}<0}.
\end{align*}
Observing that \(V_{XXd}(d,X|d\leq 4q, X-d\geq4q)=\frac{4q^2}{(X-d)^{5/2}(X-d-4q)^3/2}>0 \),
we can compute as lower bound for 
\begin{align*}
V_{XX}(d^0(X),X)&=\frac{1}{24q}\left(3\left(\sqrt{\frac{X-d}{X-d-4q}}-\sqrt{\frac{X}{X-4q}}\right)+6\left(\sqrt{\frac{X-d-4q}{X-d}}- \sqrt{\frac{X-4q}{X}}\right) \right.\\
&\left.+\left(\frac{X-4q}{X}\right)^{3/2}- \left(\frac{X-d-4q}{X-d}\right)^{3/2}\right)  \geq V_{XX}(d=0,X)=0
\end{align*}
implying that $\mathrm{d}^2V(d^0(X),X)/(\mathrm{d} X^2)\geq0$.
\end{proof}

\begin{lemma}\label{lem:URconcaveBound}
Assume $X \in [4q,8q]$, then $\mathrm{d}^2 U_R(d=X/2; X)/(\mathrm{d}X)^2<0$.
\end{lemma}
\begin{proof}
Take the case of the boundary solution: we are analyzing a one-dimensional optimization problem with respect to $\rho$. Denote the objective $f(\rho;X)$ and the optimal value by $\varphi(X)=\max_\rho f(\rho;X)$. Then, the optimal $\rho$ solves $f_\rho=0$. We obtain 
	\begin{align*}
		\varphi'(X)&=\underbrace{f_\rho}_{=0 \text{ by optimality}} \rho'(X) + f_X \\
		\varphi''(X)&=  \underbrace{f_\rho}_{=0\text{ by optimality}} \rho''(X)  +\underbrace{(f_{\rho \rho} \rho'(X)+ f_{X\rho})}_{=0 \text{ by total differentiation of FOC}} \rho'(X) + f_{XX} + \rho'(X) f_{X\rho}\\
		&=f_{XX}-\frac{f_{X\rho}^2}{f_{\rho\rho}} =\rho^\eta(X) V_{XX}(X/2;X)+\frac{(V_X-\frac{V}{X})^2}{V\frac{c''}{c'}}
	\end{align*}
	where we used that $\sigma^2_{XX}(X/2;X)=0$. The expression yields as condition for the value to be strictly concave $\rho^\eta(X) c''/c'>-(V_X-\frac{V}{X})^2/(V_{XX}V)$
	where the inequality sign changed direction as $V_{XX}(X/2;X)<0$ by \cref{lem:secondtotalderivative} in the region considered. Further, note that the left-hand side larger than two by the properties of the inverse error function. We will show that the right-hand side is below one which is sufficient for concavity. We therefore consider the right-hand side at the boundary solution, which simplifies to
	\begin{equation*}
		\frac{\sqrt{X-4q}\left(X^{3/2}-2(X+2q)\sqrt{X-4q}\right)^2}{4\left(X^2 - 2q^2 - 2qX \right)\left(X^{3/2}-2(X-4q)\sqrt{X-4q}\right)}.
	\end{equation*}
	We now show that it is also smaller than one. Because the denominator is positive, a necessary and sufficient condition is
	\begin{equation*}
		\resizebox{\textwidth}{!}{$\sqrt{X-4q}\left(X^{3/2}-2(X+2q)\sqrt{X-4q}\right)^2 - 4\left(X^2 - 2q^2 - 2qX \right)\left(X^{3/2}-2(X-4q)\sqrt{X-4q}\right)<0.$}
	\end{equation*}
	Factoring out $\sqrt{\frac{X-4q}{X}}$, dividing by $X^{3/2}$, and simplifying the condition becomes
	\begin{equation}\label{eq:randomcondition}
		\sqrt{\frac{X-4q}{X}} \left(13 X^2 - 48q X\right) -8 X^2+ 16 q X +40 q^2<0.
	\end{equation}
	Notice that $\sqrt{\frac{X-4q}{X}}$ increases in $X$ and thus attains its upper bound for $X=8q$ at $1/\sqrt{2}$. Moreover, since $13 X>48q $ for any $X \in [4q,8q]$, this implies that \eqref{eq:randomcondition} holds if
	\begin{equation*}
		\begin{split}
		\frac{13 X^2 - 48q X}{\sqrt{2}} -8 X^2+ 16 q X +40 q^2 &<0 \\
		\left(\frac{13}{\sqrt{2}} - 8\right) X^2 + (16-\frac{48}{\sqrt{2}})q X + 40 q^2 &<0
		\end{split}
	\end{equation*}
	Now notice that $13/\sqrt{2}>8$ and thus the LHS is strictly convex meaning a maximum must be at one of the boundaries in $X$. But then at $X=4q$ we have
	\begin{equation*}
				\left(\frac{13}{\sqrt{2}} - 8\right) 16 q^2 + (16-\frac{48}{\sqrt{2}}) 4q^2 + 40 q^2 =  8(\sqrt{2}-3)q^2<0
	\end{equation*}
	and at $X=8q$
	\begin{equation*}
		\left(\frac{13}{\sqrt{2}} - 8\right) 64 q^2 + (16-\frac{48}{\sqrt{2}}) 8 q^2 + 40 q^2 =8(28\sqrt{2}-43)q^2<0.
\end{equation*}
which implies that condition \eqref{eq:randomcondition} holds, thus closing the proof.
\end{proof}

\begin{lemma}\label{URconvexIN}
Let $d^i<X/2$ be a local maximum of $u_R(\rho,d,X)$. If $d^i(X)$ exists on $X \in [4q,8q]$, then $\mathrm{d}^2 U_R(d=d^i(X); X)/(\mathrm{d}X)^2>0$.
\end{lemma}

\begin{proof}
The implicit function theorem yields for $d'(X)$ and $\rho'(X)$
	\begin{align*}
		\left(\begin{matrix}d'(X) \\ \rho'(X)\end{matrix}\right) = -\frac{1}{f_{dd}f_{\rho\rho}-f_{\rho d}^2}\left(\begin{matrix} f_{dX}f_{\rho\rho} - f_{\rho X}f_{d\rho} \\ f_{\rho X}f_{dd}-f_{dX}f_{d\rho}\end{matrix}\right).
	\end{align*}
	Note that $-\frac{1}{f_{dd}f_{\rho\rho}-f_{\rho d}^2}<0$ as this is $-\frac{1}{det(\mathcal{H})}$ and the determinant of the second principal minor being positive is a necessary second order condition for a local maximum given that the first ($f_{\rho \rho}$) is negative. 

	Denote the objective $f(\rho,d;X)$ and the optimal value by $\varphi(X)=\max_{\rho,d} f(d,\rho;X)$. Then, the optimal $(d,\rho)$ solves $f_\rho=0$ and $f_d=0$. Differentiating the value of the researcher twice with respect to $X$ yields
	\begin{align*}
		\varphi'(X)&=\underbrace{f_\rho}_{=0 \text{ by optimality}} \rho'(X)+\underbrace{f_d}_{=0 \text{ by optimality}}  d'(X)+f_X \\
		\varphi''(X)&=  \underbrace{f_\rho}_{=0\text{ by optimality}} \rho''(X)+\underbrace{f_d}_{=0\text{ by optimality}} d'(X) \\
		&+d'(X)\underbrace{\left(f_{dX}+f_{dd}d'(X)+f_{d\rho}\rho'(X)\right)}_{=0\text{ by total differentiation of foc wrt }d} +\rho'(X)\underbrace{\left(f_{\rho X}+f_{\rho d}d'(X)+f_{\rho\rho}\rho'(X)\right)}_{=0\text{ by total differentiation of foc wrt }\rho} \\
		&+f_{dX}d'(X)+f_{\rho X}\rho'(X)+f_{XX} =f_{dX}d'(X)+f_{\rho X}\rho'(X)+f_{XX}.
	\end{align*}
	Observe first that $f_{XX}>0$ as $f_{XX}=\rho V_{XX}(d;X)-\eta \tilde{c}(\rho)\sigma^2_{XX}(d;X)$ and $V_{XX}>0$ by proof of \Cref{cor:maxinteriorsmaller8} (in particular, \cref{lem:secondtotalderivative}) and $\sigma^2_{XX}(d;X)=-\frac{2d^2}{X^3}$. Next, we show $f_{dX}d'(X)+f_{\rho X}\rho'(X)>0$ using the implicit function theorem together with the property of the local maximum that $f_{\rho \rho}f_{dd}>f_{\rho d}^2$.
	\begin{align*}
		f_{dX}d'(X)+f_{\rho X}\rho'(X) &= -f_{dX} \left(\frac{f_{dX}f_{\rho\rho} - f_{\rho X}f_{d\rho}}{f_{dd}f_{\rho\rho}-f_{\rho d}^2}\right) - f_{\rho X}\left(\frac{ f_{\rho X}f_{dd}-f_{dX}f_{d\rho}}{f_{dd}f_{\rho\rho}-f_{\rho d}^2} \right).
	\end{align*}
	As we only need to sign this expression, we can ignore the denominator to verify
	\begin{align*}
		-f_{dX}(f_{dX}f_{\rho\rho} - f_{\rho X}f_{d\rho})-f_{\rho X}(f_{\rho X}f_{dd}-f_{dX}f_{d\rho})>0 \Leftrightarrow
		\frac{f_{dX}}{f_{\rho X}}\frac{f_{\rho\rho}}{f_{d\rho}} + \frac{f_{\rho X}}{f_{dX}}\frac{f_{dd}}{f_{\rho d}}>2.
	\end{align*}
	where we used the signs of the terms that follow because 

	\noindent\begin{minipage}{.5\linewidth}
	\begin{align*}
		f_{\rho \rho}&=-\eta \tilde{c}_{\rho \rho}(\rho)\sigma^2<0\\
		f_{dX}&=\rho V_{dX}-\eta \tilde{c}(\rho) \sigma^2_{dX}<0
		\end{align*}
		\end{minipage} \hspace{.1cm}
		\begin{minipage}{.5\linewidth}
		\begin{align*}
		f_{d\rho}&=V_d-\eta \tilde{c}_{\rho}(\rho) \sigma^2_d <V_d-\eta \frac{\tilde{c}(\rho)}{\rho}\sigma^2_d=0\\
		f_{\rho X}&=V_X-\eta \tilde{c}_{\rho}(\rho)\sigma^2_{X}  < V_X-\eta \sigma^2_X \tilde{c}(\rho)/\rho <0
	\end{align*}
	\end{minipage}

	which in turn follow from the first-order conditions and \Cref{cor:maxinteriorsmaller8}.

	Because $f_{\rho \rho}f_{dd}-f_{\rho d}^2>0$, we can replace $\frac{f_{\rho\rho}}{f_{d\rho}}$ with $\frac{f_{d\rho}}{f_{dd}}$ as $\frac{f_{\rho\rho}}{f_{d\rho}}>\frac{f_{d\rho}}{f_{dd}}$ yielding
	\begin{align*}
		2&<\frac{f_{dX}}{f_{\rho X}}\frac{f_{d\rho}}{f_{dd}} + \frac{f_{\rho X}}{f_{dX}}\frac{f_{dd}}{f_{\rho d}}
	\end{align*}
	which is true as the right-hand side can be written as $g(a)=a+\frac{1}{a}$. 
	$g(a)$ is a strictly convex function for $a>0$ and minimized at $a=1$ with $g(a=1)=2$. 
\end{proof}

\begin{lemma}\label{homoq}
	$d^\eta(\infty)$ is linear in $q$ and $\rho^\eta(\infty)$ is constant in $q$. 
\end{lemma}
\begin{proof}
The lemma follows because $\sigma^2(mq;\infty)=mq$ and thus (by \cref{prop:value_knowledge}) the functions $f(m,q):=V(mq;\infty)/\sigma^2(mq;\infty)$ and $g(m,q):=V_d(mq;\infty)$ are homogeneous of degree $0$ in $q$. It is then immediate from \eqref{eq:FOCd} and \eqref{eq:FOCrho} that $d^\eta(\infty)$ is homogeneous of degree 1 in $q$ and $\rho^\eta(\infty)$ is homogeneous of degree $0$. Noticing that $d^\eta(\infty)(q=0)=0$ implies the result.
\end{proof}

\setlength\bibitemsep{0.5\itemsep}
\printbibliography
\end{refsection}

\end{document}